\documentclass[journal,twoside,web]{ieeecolor}
\usepackage{multirow}
\usepackage{float}
\usepackage{generic}
\usepackage{cite}
\usepackage{graphicx}
\usepackage{soul} 
\usepackage{xcolor}
\usepackage{tikz} 
\usepackage{color}
\usepackage{graphicx}
\graphicspath{{figures/}{./}}
\usepackage{textcomp}
\usepackage{trimclip}
\usepackage[mathscr]{euscript}
\usepackage{array}
\usepackage{eqparbox}
\usepackage{url}
\usepackage{relsize}
\usepackage{dsfont}
\usepackage{nopageno}
\usepackage{graphicx}
\usepackage{amsmath,amsfonts,mathrsfs,bm}
\usepackage[mathscr]{euscript}
\usepackage{enumerate}
\usepackage{tikz-cd}
\usepackage{algorithm}
\usepackage{algpseudocode}

\newtheorem{theorem}{Theorem}
\newtheorem{definition}{Definition}
\newtheorem{lemma}{Lemma}
\newtheorem{remark}{Remark}
\newtheorem{proposition}{Proposition}
\newtheorem{corollary}{Corollary}

\newtheorem{problem}{Problem}

\newcommand{\ess}{\mathrm{ess}}
\usepackage{lipsum} 
\usepackage{booktabs}

\usepackage{hyperref}
\hypersetup{hidelinks=true}
\usepackage{textcomp}
\def\BibTeX{{\rm B\kern-.05em{\sc i\kern-.025em b}\kern-.08em
		T\kern-.1667em\lower.7ex\hbox{E}\kern-.125emX}}
	
\begin{document}

	\title{Diversity and Interaction Quality of a Heterogeneous Multi-Agent System Applied to a Synchronization Problem }
	
	\author{Xin Mao, Dan Wang, Wei Chen, Li Qiu, \IEEEmembership{Fellow, IEEE}
		\thanks{This work was partially supported by the Research Grants Council of Hong Kong Special Administrative Region, China, under the General Research
			Fund 16203922.}
		\thanks{X.\ Mao is with the Department of Electronic and Computer Engineering,
			The Hong Kong University of Science and Technology, Clear Water Bay, Kowloon, Hong Kong S.A.R., China. E-mails: xmaoaa@connect.ust.hk.}
		\thanks{D. Wang is with the Division of Decision and Control Systems, KTH Royal Institute of Technology, Stockholm, Sweden. Email: danwang@kth.se.}
		\thanks{W. Chen is with the Department of Mechanics and Engineering Science,
			Peking University, Beijing 100871, China. Email: w.chen@pku.edu.cn.}
   \thanks{L. Qiu is with the School of Science and Engineering, The Chinese University of Hong Kong, Shenzhen, China. Email: qiuli@cuhk.edu.cn.}
	}
	
	\maketitle

	\begin{abstract}
		In this paper, scalable controller design to achieve output synchronization for a heterogeneous discrete-time nonlinear multi-agent system is considered. The agents are assumed to exhibit potentially nonlinear dynamics but share linear common oscillatory modes. In a distributed control architecture, scalability is ensured by designing a small number of distinguished controllers, significantly fewer than the number of agents, even when agent diversity is high. Our findings indicate that the number of controllers required can be effectively determined by the number of strongly connected components of the underlying graph. The study in this paper builds on the recently developed phase theory of matrices and systems. First, we employ the concept of matrix phase, specifically the phase alignability of a collection of matrices, to quantify agent diversity. Next, we use matrix phase, particularly the essential phase of the graph Laplacian, to evaluate the interaction quality among the agents. Based on these insights, we derive a sufficient condition for the solvability of the synchronization problem, framed as a trade-off between the agent diversity and the interaction quality. In the process, a controller design procedure based on Lyapunov analysis is provided, which produces low gain, component-wise synchronizing controllers when the solvability condition is satisfied. Numerical examples are given to illustrate the effectiveness of the proposed design procedure. Furthermore, we consider cases where the component-wise controller design problem is unsolvable. We propose alternative strategies involving the design of a small inventory of controllers, which can still achieve synchronization effectively by employing certain clustering methods to manage heterogeneity.		
	\end{abstract}
	
	
	\section{Introduction}
	Synchronization is a phenomena that the behaviours of coordinated agents converge to a common trajectory over time. This process is ubiquitous in both nature and engineering, with wide-ranging applications, particularly in  power systems \cite{dorfler2013synchronization}, biological systems \cite{strogatz1990biology}, parallel/distributed computing and social network \cite{arenas2008synchronization}. The extensively studied topics including consensus, flocking, swarming, formation, rendezvous in recently years can be unified in the output synchronization framework \cite{cucker2007emergent,jadbabaie2003coordination,ren2005consensus}.
	It is noteworthy that the majority of works in the literature focuses on continuous-time multi-agent systems. However, many practical applications naturally involve problems in a discrete-time setup. For example, early studies on the consensus problem originated from the field of distributed computing in the computer science community. In this context, distributed processors collectively solve a computational task, with each processor storing a discrete-time variable that evolves according to a specific iterative strategy \cite{bertsekas1989parallel}. In social networks, the seminal work \cite{degroot1974reaching} introduces a discrete-time updating rule, now known as the DeGroot model, to describe opinion dynamics. Here, a group of individuals act together as a team, each holding an opinion represented by a probability density function. After communication, beliefs are updated using a row-stochastic opinion update matrix in the next step. The investigation focuses on whether consensus can be reached on each issue. Additionally, the well-known Vicsek model is a simple yet powerful discrete-time model in the flocking literature \cite{vicsek1995novel}. In this model, each agent moves at the same speed but with different headings, following a nearest neighbor interaction rule. Over time, all agents align and move in the same direction. These examples motivate the study of discrete-time complex networks.
	
	Early research on synchronization in multi-agent systems often assumed that the dynamics of agents were either first-order or second-order, both in continuous and discrete time domains \cite{jadbabaie2003coordination, ren2008distributed}. Initial studies focused on consensus under fixed directed topologies, where graph theory played a pivotal role. A well-known result from this line of research is that a necessary and sufficient condition for the consensus of a group of agents modeled by a single integrator is the presence of a spanning tree in the communication graph \cite{olfati2007consensus}. However, practical applications often involve scenarios far more complex and less idealized. Constraints in communication or sensing significantly increase the control challenges. As a result, researchers began to account for more intricate edge dynamics, including linear time-invariant (LTI) dynamic topologies, time-varying topologies, noise, time delays, and packet drops in communication channels \cite{ olfati2007consensus, qi2016mas, ren2005consensus}.

Another significant research direction involves investigating more complex agent dynamics from a dynamic systems perspective. Early studies focused on homogeneous networks, where agents share identical LTI dynamics and properties \cite{scardovi2009synchronization, you2011network}. For such networks, employing a uniform controller across all agents is both natural and efficient. This approach simplifies control
design and analysis by uniformly applying a single controller to all agents, ensuring consistent behavior and facilitating synchronization. Prior work has demonstrated that a low-gain constant controller can effectively solve the synchronization problem for discrete-time integrators \cite{olfati2007consensus, ren2008distributed}, and the homogeneous case has been explored with uniform controllers \cite{ gu2011consensusability,li2009consensus,you2011network}. Over the last decade, however, attention has increasingly shifted toward heterogeneous networks, where agents exhibit diverse dynamics, including time-varying and nonlinear behaviors \cite{alvergue2015output, arcak2007passivity, lestas2006scalable,pates2016scalable}. This shift highlights the complex interplay between agent dynamics, interaction protocols, and controllers, especially in general networks. It is well established that the synchronization problem can be framed as a simultaneous stabilization problem \cite{ren2008distributed}, enabling the application of various stability theory techniques. Given the large-scale nature of these networks, scalability is a critical consideration in control strategy design, which is often distributed due to the limited availability of global information.

In making a large-scale heterogeneous network work properly under a small number of control actions, we have the intuition that the difficulty relies on the degree of heterogeneity of the agents. We call such a degree the agent diversity. We also have the intuition that the quality of interaction of the underlining graph contributes to the difficulty. In this paper, we will precisely define such concepts using recent advances in matrix phase analysis \cite{wd2019,wd2023}. Our main result gives a solvability condition to the synchronizing controller design problem mentioned above, which shows a simple trade-off between the agent diversity and interaction quality.

The use of phase analysis in multi-agent network synchronization analysis and synthesis has been reported in \cite{ mao2024robust,wang2024synchronization}. Building on these foundations, we investigate networks with possibly nonlinear agent dynamics and derive quantitative phase-type conditions for synchronizing controller design. The nonlinear dynamics are assumed to be linear in oscillatory modes and nonlinear in transient modes. Our primary focus is on designing component-wise controllers based on the network topology, where the number of controllers is equal to the number of strongly connected components of a graph. This is in contrast with the common practice in the literature of adopting agent-based controllers where each agent calls for a specially designed controller. Such a practice has no scalability but excellent solvability. This is also in contrast to a uniform control approach \cite{wang2024synchronization} where an identical controller is utilized for all agents. Such an approach tends to have a too strong scalabilty with sacrified solvability. The component-wise controller approach is motivated by our recent findings that demonstrate its effectiveness in achieving a balance between scalability and performance. This hierarchical strategy enables the assignment of controllers to different strongly connected components, simplifying the control design while maintaining the flexibility to accommodate agent diversity.

 In this paper, we pay great attention to the solvability issue, i.e., whether a suitable set of controllers can be designed to enforce synchronization. When the solvability condition is satisfied, we provide a systematic design procedure for constructing viable component-wise controllers. Component-wise controllers offer the additional benefit of addressing a certain level of agent diversity by employing robust control techniques that focus on worst-case scenarios. To further enhance the balance of scalability and solvability, in case when the component-wise controller architecture does not suffice, we propose to partition the agents into clusters based on our physical understanding and even possibly artificial intelligence clustering techniques and use one dedicated specific controller for each intersection of a cluster and a component. This component and cluster combined approach provides an effective way to manage heterogeneity in scenarios where component-wise control is inadequate.

The remainder of the paper is organized as follows. Section~\ref{sec:preliminaries} provides the necessary background and preliminaries on graph theory and nonlinear system stability. The synchronization problem formulation is presented in Section~\ref{sec:synformulation}. In Section~\ref{sec:diversity}, we introduce matrix diversity and network symmetry, focusing on aspects related to matrix phases, simultaneous alignment, and network symmetry. Results on synchronization synthesis are discussed in Section~\ref{sec:synthesis}, covering both component-wise and component-cluster combined controller designs. Section~\ref{sec:simulation} provides simulation results, and the paper concludes in Section~\ref{sec:conclusion}.
	
Notation used in this paper is mostly standard. Let $\mathbb{R}$ and $\mathbb{C}$ be the set of real and complex numbers, respectively. Denote the set of integers by $\mathbb{Z}$. For a matrix $A\in\mathbb{C}^{n\times n}$, $A^*$ denotes its complex conjugate transpose. Let $\mathcal R(A)$ be the range of $A$. Let $\lambda(A)$ and $\angle\lambda(A)$ be the sets of eigenvalues and their angles of $A$. The Kronecker product of two matrices $A$ and $B$ is denoted by $A\otimes B$. The identity matrix is denoted by $I$. For a vector $x\in\mathbb{C}^n$, $x^*$ denotes its complex conjugate transpose. We use $\mathbf{1}$ to denote the vectors with all entries equal to $1$. The Euclidean norm is denoted by $\| \cdot \|_2$. Let $\mathcal{R}^{m\times m}$ be the set of $m\times m$ real rational transfer matrices and let $\mathcal{RH}_\infty^{m\times m}\subset\mathcal{R}^{m\times m}$ contain all its proper stable elements. In this paper, we will adopt the $z$-transform in discrete time. Therefore, $\mathcal{RH}_\infty^{m\times m}$ is the set of real rational transfer matrices with poles in the open unit disk. This is different from the definition in the complex function theory in which $\mathcal{RH}_\infty$ usually means the set of real rational functions analytic on the closed unit disk.

	\section{Preliminaries}\label{sec:preliminaries}
	\subsection{Graph theory}
	Consider a directed graph $\mathcal{G}=(\mathcal{V}, \mathcal{E})$ with a set of vertices $\mathcal{V}=\{v_1, \dots, v_n\}$ and a set of directed edges $\mathcal{E}\subseteq\mathcal{V}\times\mathcal{V}$. A sequence of edges $(v_1,v_2), (v_2,v_3),\dots,(v_{k-1},v_k)$ with $(v_{j-1},v_j)\in\mathcal{E}$, $j=\{2,\dots,k\}$ is called a directed path from node $v_1$ to node $v_k$. Here, we assume that the directed graph does not have self-loops. A directed graph $\mathcal{G}$ is said to be strongly connected if every vertex has paths to every other vertex. It is said to have a spanning tree if at least one node exists, called a root, that has directed paths to all other nodes. Obviously, the existence of a spanning tree is a weaker condition than being strongly connected. The graph $\mathcal{G}$ is called an undirected graph if all the edges are bidirectional. For an undirected graph, the strongly connectedness is simply termed connectedness. Furthermore, the existence of a spanning tree is equivalent to being connected. A weighted graph is a graph in which a weight is assigned to each edge. A strongly connected graph is (weight) balanced if for each node, the total coming weights are equal to the total leaving weights.
	
	The weighted adjacency matrix is defined as $A=[a_{ij}]$, where $a_{ji}$ is a positive real number (the weight) if $(v_i, v_j)\in\mathcal{E}$ and $a_{ji}=0$ otherwise. The indegree matrix is defined as $D=\text{diag}\{\sum_{j=1}^n a_{1j}, \dots, \sum_{j=1}^n a_{nj}\}$. The Laplacian matrix is defined as $L=D-A$, whose row sums are equal to zero. A Laplacian matrix always has a zero eigenvalue with a corresponding right eigenvector $\mathbf{1}_n$. A necessary and sufficient condition for $\mathbf{1}_n$ being also a left eigenvector corresponding to zero eigenvalue is that the graph is balanced \cite{olfati2007consensus}. Furthermore, zero is a simple eigenvalue if and only if the graph has a spanning tree \cite{ren2005consensus}. The Laplacian matrix of a strongly connected graph is irreducible, i.e., not similar via permutation to a block upper triangular matrix. 

\subsection{Nonlinear systems and feedback stability} 
A signal of $m$-dimension is a bilateral sequence 
\[
u=\{\dots,u(-2),u(-1), \mid u(0), u(1), u(2), \dots \}
\]
where $u(t)\in\mathbb{R}^m$. The underlying time axis is the set of integers $\mathbb{Z}$. The vertical line marks the zero-time index, providing a reference point for the sequence. The set of all signals is denoted by $\ell(\mathbb{Z})$. A signal $u\in\ell(\mathbb Z)$ is said to belong to $\ell_2(\mathbb Z)$ if 
\[
\sum_{n=-\infty}^{\infty}\|u(t)\|_2^2<\infty.
\]
The causal subspace $\ell_2(\mathbb Z_+)$ is defined as
\[
\ell_2(\mathbb Z_+)=\left\{u\in\ell_2(\mathbb Z): u(t)=0\ \mathrm{for}\ t<0\right\}.
\]
	
Consider a discrete-time dynamical system characterized by the input-output relation 
	\[
	y=\bm Pu,
	\]
	where $\bm P$ is operator mapping signals from one $\ell_2(\mathbb Z_+)$ space to another. It is possibly nonlinear. For $T\in\mathbb{R}$, define the truncation $\mathbf \Gamma_T$ on all $u$ by
 \begin{equation*}
(\mathbf\Gamma_Tu)(t)=\begin{cases}u(t),\ &t\leq T,\\
0,\ &t>T.
\end{cases}
 \end{equation*}
 The operator $\bm P$ is said to be causal if $\mathbf\Gamma_T\mathbf P = \mathbf\Gamma_T\mathbf P\mathbf\Gamma_T$
for all $T\in\mathbb Z$ and is said to be noncausal if it is not causal. Assume that $\bm P$ is a causal operator, and it always maps the zero signal to itself, i.e., $\bm P0=0$. In addition, we restrict our attention to nontrivial systems with an equal number of inputs and outputs.

The $\mathcal H_\infty$ norm of the system $\bm P$ is defined as
   	\[
    \|\bm P\|_\infty=\sup_{u\in \ell_2(\mathbb Z_+), u\neq 0}\frac{\|\bm Pu\|_2}{\|u\|_2}.
    \]	
    This gain measures the maximum amplification of the energy of any input signal by the system. Here we consider the input-output stability of the system.
	A system $\bm P$ is said to be stable if $\|\bm P\|_\infty<\infty$.

	Consider a standard closed-loop system with a negative feedback configuration as illustrated in Fig.~\ref{fig:closedsys}, where $\bm P$ and $\bm C$ represent stable nonlinear systems. The feedback interconnection of these systems, denoted as $\bm P\#\bm C$, is said to be stable if the inverse of $\begin{bmatrix}\bm I & \bm C\\
	-\bm P & \bm I  \end{bmatrix}$ exists and is both causal and stable, where $\bm I$ denotes the identity system.
	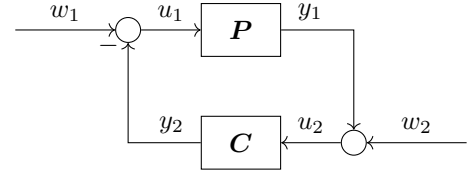
\begin{figure}[htbp]
		\centering
		\tikzstyle{block} = [draw, rectangle,
		minimum height=2em, minimum width=3em]
		\tikzstyle{sum} = [draw, circle, node distance=1.5cm]
		\tikzstyle{input} = [coordinate]
		\tikzstyle{output} = [coordinate]
		\tikzstyle{pinstyle} = [pin edge={to-,thin,black}]
		
		\begin{tikzpicture}[auto, node distance=1.5cm]
			\node [input, name=input] {};
			\node [sum, right of=input] (sum1) {};
			\node [block, right of=sum1] (P) {$\bm P$};
			\node [block, below of=P] (C) {$\bm C$};
			\node [sum, right of=C] (sum2) {};
			\node [output, right of=sum2] (output) {};
			\draw [draw,->] (input) -- node {$w_1$} (sum1);
			\draw [->] (sum1) -- node {$u_1$} (P);
			\draw [->] (P) -| node [pos=0.2] {$y_1$} (sum2);
			\draw [<-] (C) -- node {$u_2$} (sum2);
			\draw [<-] (sum2) -- node {$w_2$} (output);
			\draw [->] (C) -| node[pos=0.99] {$-$}
			node [pos=0.2, above] {$y_2$} (sum1);
		\end{tikzpicture}
		\caption{\label{fig:closedsys}Negative feedback interconnection of $\bm P$ and $\bm C$.}
	\end{figure}
 A fundamental result in robust control theory, particularly within the context of $\mathcal H_\infty$ control, is the nonlinear small gain theorem \cite{van2000l2,desoer2009feedback}. This theorem provides a sufficient condition for the stability of the closed-loop system $\bm P\#\bm C$.
	\begin{lemma}[Nonlinear Small Gain Thoerem]
		The closed-loop system is stable if $\|\bm P\|_\infty\|\bm C\|_\infty<1$.
	\end{lemma}
This theorem underscores the critical importance of ensuring that the product of the system gains remains below unity to guarantee the stability of an interconnected system. It has become a foundational principle in the analysis and design of robust control systems, significantly influencing the development of $\mathcal H_\infty$ control theory and other advanced control strategies. The application of this theorem is particularly useful in our work, where maintaining robust stability in the face of uncertainties is paramount.

\section{Problem formulation}\label{sec:synformulation}
Consider a complex discrete-time dynamic network of heterogeneous agents. Each agent $\bm P_i$ consists of two components: a persistent mode $\bar{\bm P}_i$ and a stable mode $\bm\Delta_i$ such that 
\begin{align}\label{eq:agent}
\bm P_i=\bar{\bm P}_i+\bm\Delta_i.    
\end{align}
Assume that the persistent plant mode $\bar{\bm P}_i$ is linear time-invariant, characterized by strictly proper rational transfer function $\bar P_{i}(z)\in\mathcal{R}^{m\times m}$, and is represented in the state space form as:
	\[
	\bar P_{i}(z)=\left[\begin{array}{c|c}
		A_i&B_i\\\hline
		C_i&0
	\end{array}\right]=C_i(zI-A_i)^{-1}B_i
	\]
  for $ i=1,2,\dots,n$, where $A_i, B_i, C_i$ are the state, input and output matrices respectively. The dimensions of the inputs and outputs of all the agents are $m$. It is assumed that $(A_i, B_i)$ are controllable, $(C_i, A_i)$ are observable for each agent.
	
The agents are considered to be semi-stable, meaning that $\bar P_{i}(z)$ is semi-stable and $\bm\Delta_i$ is stable. Assume these agents share common persistent internal modes, i.e., the eigenvalues of $A_i$ are the same and lie on the unit circle, allowing them to autonomously generate the same common oscillatory outputs. We denote this set of modes by $e^{j\Omega}=\{e^{j0},e^{\pm j\omega_1},\dots, e^{\pm j\omega_q}, e^{j\pi}\}$, with $0<\omega_1<\dots<\omega_q<\pi$. Assume that for each persistent mode, its geometric multiplicity and algebraic multiplicity are the same, which is equal to $m$. Despite the commonality, the agents may differ significantly in their stable modes and system orders.   The partial fractional expansion of $\bar P_i(z)$ can be written in the form
	{\begin{multline}
			\bar P_i(z)=\frac{N_{0i}}{z-1}+\frac{N_{\pi i}}{z+1}+\frac{N_{1i}}{z-e^{j\omega_1}}+\frac{\bar{N}_{1i}}{z-e^{-j\omega_1}}\\+\cdots+\frac{N_{qi}}{z-e^{j\omega_q}}+
			\frac{\bar{N}_{qi}}{z-e^{-j\omega_q}}, \label{eq: dtagents}
	\end{multline}}%
	\normalsize
	where $N_{0i}, N_{\pi i}\in\mathbb{R}^{m\times m}$ are the residues of $\bar P_i(z)$ at the pole $1$ and $-1$, $N_{li}\in\mathbb{C}^{m\times m}$ for $l=1,\dots,q,$ are the residues of $\bar P_i(z)$ at the pole $e^{j\omega_l}$.
	
Let $u_i(t)\in\mathbb{R}^m$ and $y_i(t)\in\mathbb{R}^m$ be the input and output of $i$-th agent, respectively. Assume that the agents in the network are coupled through a diffusive interaction. The dynamic behavior of the agents is given by $y_i=\bm P_iu_i$. Here $u_i=w_i+v_i$, where $w_i$ is a bias function accounting for the initial conditions of the agents and $v_i$ is the control input  determined by the relative output differences between connected agents. The input $v_i$ is captured by
	\[
v_i=\sum_{(i,j)\in\mathcal{E}}a_{ij}\bm C_i(y_j-y_i),
	\]
where $\bm C_i$ is the local controller assigned to agent $i$ and $a_{ij},i,j=1,\dots,n$, are the nonnegative edge weights of the underlying graph of the network, capturing the underlying graph's topology and interaction strengths. The weights $a_{ij}$ are assumed to be known {\em a priori}. Denote the Laplacian matrix of the directed graph with the static nonnegative edge weights $a_{ij}$ by $L$. Assume the controller $\bm C_i$ is linear time-invariant, with the transfer function given by $C_i(z)$.
A particular case of interest is when $\bm C_i$ is a static controller. The block diagram representing the system is illustrated in  Fig. \ref{fig: directednetwork}. Assuming the directed graph has a spanning tree, which ensures minimal connectivity, our goal is to achieve output synchronization in the network. The multi-agent system is said to reach output synchronization if  $\lim\limits_{t\rightarrow\infty}(y_i(t)-y_j(t))=0,\forall i,j\in\{1,2,\dots,n\}$ and all initial conditions. In other words,
	\[
	\lim\limits_{t\rightarrow\infty}(y_i(t)-y_0(t))=0, \forall i=1,\dots,n,
	\]
	where $y_0(t)$ is the synchronized output trajectory. 
 The synchronization synthesis problem is formulated as follows.
 \begin{problem}\label{problem1}
 Design controllers $\bm C_i$ such that the outputs of all the agents \eqref{eq:agent} asymptotically synchronize regardless of initial conditions.
 \end{problem}
	\begin{figure}[htbp]
		\begin{center}			\includegraphics[width=8cm]{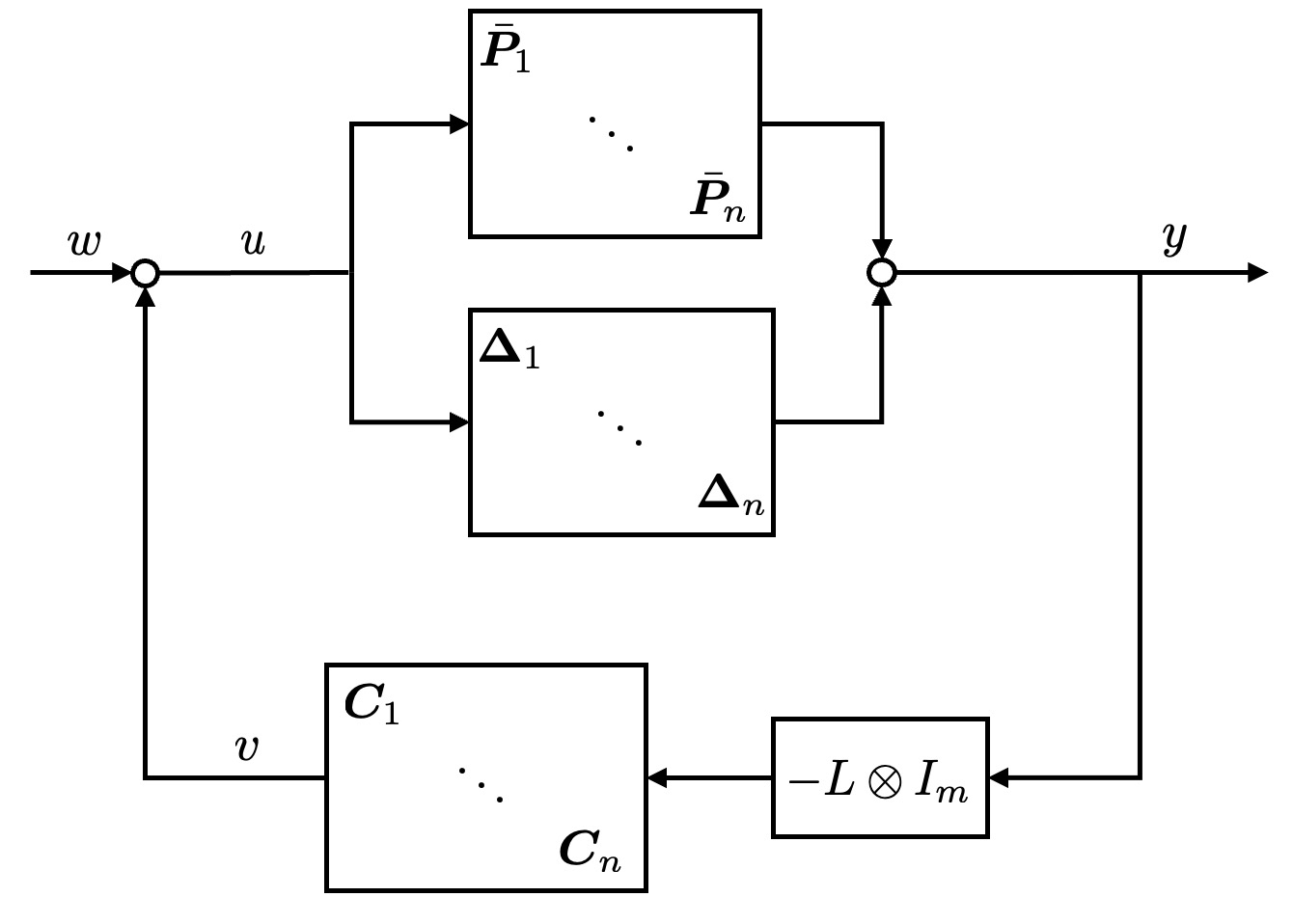}    \caption{\label{fig: directednetwork}Block diagram of a directed graph with agent-based controllers.}
		\end{center}
	\end{figure}
	
 To address the synchronization problem, define the dynamic Laplacian matrix $\bm L$ as follows:
	\begin{align*}
		\bm L_{ij}=\begin{cases}
			-a_{ij}\bm C_{i}& i\neq j,\\
			\sum_{j\neq i}a_{ij}\bm C_{i}& i=j.
		\end{cases}
	\end{align*}
 Consequently, the dynamic Laplacian matrix can be express as  $\bm L=\bm C(L\otimes I_m)$, where $\bm C=\mathrm{diag}\{\bm C_1,\dots,\bm C_n\}$.
	By defining the vectors $w=\begin{bmatrix}w_1' &\cdots &w_n'\end{bmatrix}'$, $v=\begin{bmatrix}v_1' &\cdots &v_n'\end{bmatrix}'$, $u=\begin{bmatrix}u_1' &\cdots &u_n'\end{bmatrix}'$ and $y=\begin{bmatrix}y_1' &\cdots &y_n'\end{bmatrix}'$, the network dynamics are written as
	\begin{align*}
		y&=\bm P(v+w), \\
		v&=-\bm Ly,
	\end{align*}
	where \begin{align*}
		\bm P&=\mathrm{diag}\{\bm P_1,\dots,\bm P_n\}.
	\end{align*}
	This leads to
	\begin{align}\label{haty}
		y=(I+\bm P\bm L)^{-1}\bm Pw.
	\end{align}
	The synchronization framework is depicted in Fig~\ref{fig: syndiagram}. To analyze synchronization, introduce the disagreement vector $e(z)=(J\otimes I_m)y(z)$, where $J=I_n-\frac{1}{n}\mathbf{1}_n\mathbf{1}'_n$. The matrix $J$ has a simple eigenvalue $0$ with a corresponding right eigenvector $\mathbf{1}_n$. 
 
	\begin{figure}[htbp]
		\centering
		\includegraphics[width=6cm]{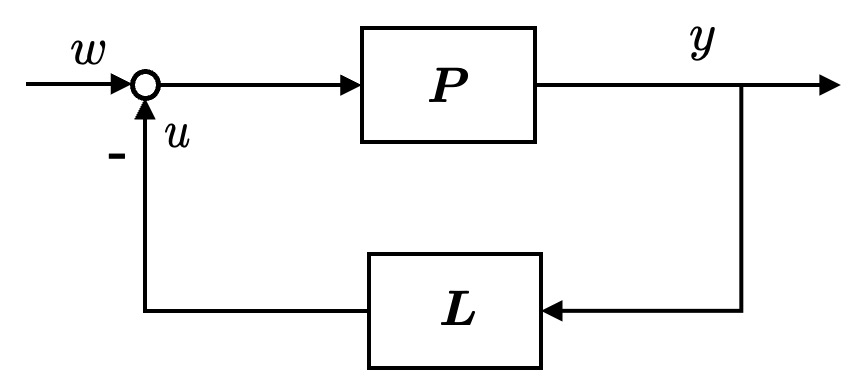}
		\caption{\label{fig: syndiagram}Block diagram of synchronization.}
	\end{figure}
	
Let $Q$ be an isometry whose columns form an orthogonal complement of $\mathrm{span}\{\mathbf{1}_n\}$. Define the matrix $U=\begin{bmatrix}Q&\frac{1}{\sqrt{n}}\mathbf{1}_n  \end{bmatrix}\otimes I_m$. The disagreement vector can then be expressed as
	\begin{align}\label{eq: e}
		e&=JUU'(I+\bm P\bm L)^{-1}UU'\bm Pw\\
		&=JU\begin{bmatrix}\bm S& 0\\ * & I\end{bmatrix}U'\bm Pw\\
		&=(Q\otimes I_m)\bm S(Q'\otimes I_m)\bm Pw,
	\end{align}
where $\bm S=(I_{nm-m}+(Q\otimes I_m)'\bm P\bm L(Q\otimes I_m))^{-1}$ and $*$ denotes irrelevant parts. Here $e$ can be treated as the tracking error of the reference signal $\bm Pw$. Achieving $\lim\limits_{t\rightarrow\infty}e(t)=0$ 
 is equivalent to achieving synchronization. The problem can be transformed to feedback stabilization problem. 

Achieving $\lim\limits_{t\rightarrow\infty}e(t)=0$ requires $S$ to be stable and the internal model of $\bm Pw$ to be contained in the loop transfer matrices. The latter is naturally satisfied since the internal model of $\bm P$ corresponds to the agent dynamics. Let 
	\begin{align*}
		\tilde{\bm P}&=(Q'\otimes I_m)(\bm P\bm C),\\ \tilde{L}&=(LQ)\otimes I_m.
	\end{align*}
	Since $Q'Q=I_{n-1}$, and $QQ'=I_n-\frac{1}{n}\mathbf{1}_n \mathbf{1'}_n$, it can be derived that $\bm S=(I+\tilde{\bm P}\tilde{L})^{-1}.$
The stability of $\bm S$ is equivalent to the stability of the feedback system shown in Fig~\ref{fig: synstability}.
	\begin{figure}[htbp]
		\centering
		\includegraphics[width=6cm]{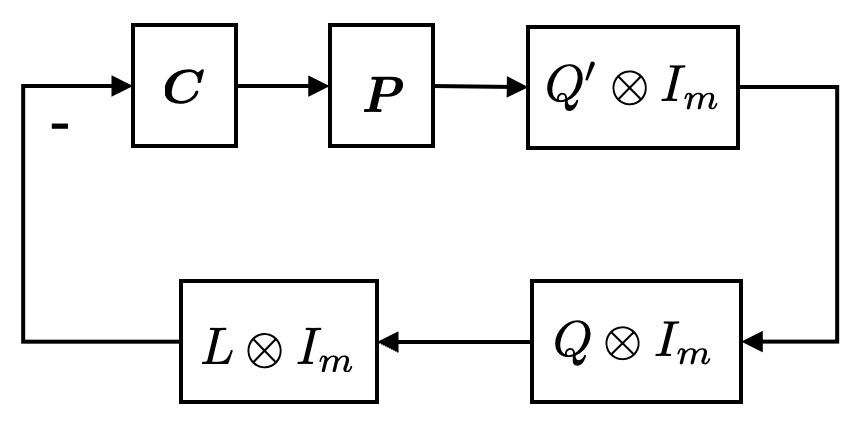}
		\caption{\label{fig: synstability}Block diagram of equivalent stability problem.}
	\end{figure}
The techniques dealing with stability problem can be naturally applied. Mathematically, allowing the distributed controller $\bm C_i$ to be different ensures that the synchronization problem is always solvable, effectively eliminating the issue of synchronizability. However, this comes with a substantial increase in design and implementation costs. Additionally, the design lacks scalability with respect to network size. In the rest of this paper, we consider the cases of component-wise controllers and component-cluster combined controllers. 


 \section{Matrix diversity and network symmetry}\label{sec:diversity}
 In this section, we propose the simultaneous alignment problem and define a measure of the diversity for a set of matrices based on the phases of complex matrices, which were defined in \cite{wd2019,wd2023}. Moreover, we review the notion of essential phase for a Laplacian matrix introduced in \cite{wd2023} and propose to use it as a measure of the network symmetry.
 
	\subsection{Matrix phases}
	Given a matrix $A\in\mathbb{C}^{m\times m}$, the numerical range of $A$ is defined to be
	\[
	W(A)=\{x^*Ax : x\in\mathbb{C}^m,\|x\|_2=1\}.
	\]
	This is a convex and compact subset of the complex plane \cite[Section 1.2]{horntopics} and contains the spectrum of $A$.
	Moreover, the angular numerical range of $A$ is defined to be
	\[
	W'(A)=\{x^*Ax : x\in\mathbb{C}^m,x\neq 0\}.
	\]
	The matrix $A$ is said to be semi-sectorial if the origin is not in the interior of $W(A)$. A semi-sectorial matrix is said to be quasi-sectorial if the origin is not on the smooth boundary of $W(A)$. Furthermore, it is said to be sectorial if the origin is not contained in $W(A)$. We call $A\in\mathbb{C}^{m\times m}$ rotationally indefinite Hermitian if there exists an $\alpha$ such that $e^{j\alpha}A$ is indefinite Hermitian. Rotationally indefinite Hermitian $A$ belongs to the class of semi-sectorial matrices.  Let $\mathrm{Re}(A)=(A+A^*)/2, \mathrm{Im}(A)=(A-A^*)/(2j)$. The matrix $A$ is said to be accretive if $\mathrm{Re}(A)$ is positive semidefinite, and is said to be quasi-strictly accretive if it is both accretive and quasi-sectorial.
	
	For a nonzero semi-sectorial matrix $A$, there exists a closed half plane which contains $W(A)$. Assume the half plane is given by $[\theta(A)-\pi/2, \theta(A)+\pi/2]$, where $\theta(A)\in[-\pi,\pi)$. The largest and smallest phases of $A$ are defined as
	\begin{align*}
		\overline\phi(A)&=\sup\limits_{ x\neq 0,\ x^*Ax\neq 0}\angle x^*Ax,\\
		\underline\phi(A)&=\inf\limits_{ x\neq 0,\ x^*Ax\neq 0}\angle x^*Ax,
	\end{align*}
    taking values in $[\theta(A)-\pi/2, \theta(A)+\pi/2]$.
	The phase interval of $A$ is defined to be
 \[\Phi(A)=[\underline\phi(A),\overline\phi(A)].\]
  The phase interval can be computed by solving linear matrix inequalities (LMIs) obtained from Lemma~\ref{lem: phaselmi}. 
\begin{lemma}\label{lem: phaselmi}
Let $\alpha\in [0,\frac{\pi}{2})$ and $A\in\mathbb C^{m\times m}$. Then $\Phi(A)\subset[-\alpha,\alpha]$ if and only if $\mathrm{Re}(A)\geq 0$ and 
\begin{align}\label{eq: tanalpha}
-\tan\alpha\ \mathrm{Re} (A)\leq\mathrm{Im} (A)\leq\tan\alpha\ \mathrm{Re} (A).
\end{align}
\end{lemma}
\begin{proof}
We first show the necessity.
Assume the rank of $A$ is $r$. Since $A$ is quasi-sectorial, it can be 
written as $A=T^*DT$, where $T$ is nonsingular and 
\begin{align*}
D=\begin{bmatrix}
0_{m-r}&  &  &    \\
 & 1+j\tan\underline\phi(A)\\
 &  &\ddots &\\
 & & & 1+j\tan\overline\phi(A)
\end{bmatrix}
\end{align*}
according to \cite{mao2022phases}. It follows that 
\begin{align*}
\mathrm{Re}(A)&=T^*\mathrm{diag}\{0_{m-r},1,\dots,1\}T,\\ \mathrm{Im}(A)&=T^*\text{diag}\{0_{m-r},\tan\underline\phi(A),\dots,\tan\overline\phi(A)\}T.
\end{align*}
Hence $\mathrm{Re}(A)\geq 0$. Since 
$$-\tan\alpha\ I\leq\text{diag}\{\tan\underline\phi(A),\dots,\tan\overline\phi(A)\}\leq\tan\alpha\ I,$$ it follows that (\ref{eq: tanalpha}) holds.

We then show the sufficiency. Since $\mathrm{Re}(A)\geq 0$, it follows that $A$ is semi-sectorial. Then, $A$ can be decomposed as
$A=T^*\text{diag}\{0_{m-r},D_1,D_2\}T
$, where
\begin{align*}
D_1&=\text{diag}\{e^{j\theta_1},\dots,e^{\theta_{t}}\},\\
D_2&=\text{diag}\left\{\begin{bmatrix}1&2\\0&1\end{bmatrix},\dots,\begin{bmatrix}1&2\\0&1\end{bmatrix}\right\},
\end{align*}
with $\theta_1,\dots,\theta_{t}\in[-\pi/2,\pi/2]$ \cite{mao2022phases}.
From \eqref{eq: tanalpha}, we have 
$$-\tan\alpha\ \mathrm{Re}(D_1)\leq\mathrm{Im}(D_1)\leq\tan\alpha\ \mathrm{Re}(D_1).$$ 
Hence $\theta_1, \dots, \theta_t \subset [-\alpha,\alpha]$. Moreover, note that $$-\tan\alpha\ \mathrm{Re}(D_2)\leq\mathrm{Im}(D_2)\leq\tan\alpha\ \mathrm{Re}(D_2)$$ can not hold with $\alpha<\pi/2$ since $W(D_2)$ is a disk centered at 1 with radius 1. Therefore, $D_2$ does not exist. The proof is completed.
\end{proof}

	The phases defined above have many nice properties. One may refer to \cite{wd2019,wd2023} for more details. Here we present several properties that are useful in later developments. The next lemma characterizes the relationship between the phases of a sectorial matrix and its compression.
	\begin{lemma}[\cite{wd2023}]\label{lem:compression}
		Let $A\in\mathbb{C}^{m\times m}$ be a nonzero semi-sectorial matrix and $\tilde{A}\in\mathbb{C}^{(m-k)\times(m-k)}$ be a nonzero compression of $A$. Then $\tilde{A}$ is semi-sectorial and
$\Phi(\tilde A)\subset\Phi(A)$.
	\end{lemma}
	Another important property is about the matrix product.
	\begin{lemma}[\cite{wd2023}]
 \label{lem:major}
		Let $A, B\in\mathbb{C}^{m\times m}$ be semi-sectorial and sectorial. Then the number of nonzero eigenvalues of $AB$ is equal to the rank of $A$, and the inequality
		\begin{align}\label{eq:matrixspt}
	\angle\lambda_i(AB)\in\Phi(A)+\Phi(B)
		\end{align}
		is satisfied if $\angle \lambda_i(AB)$ take values in $(\theta(A)+\theta(B)-\pi, \theta(A)+\theta(B)+\pi)$.
	\end{lemma}

	\subsection{Simultaneous alignment and matrix diversity}

In various applications, it is desirable that phases of a set of matrices are similar or even contained in the same interval. However, in practice, this is not often the case, and some matrices may not even be semi-sectorial. To overcome this, we introduce a matrix $K$ to transform the set of matrices so that they become semi-sectorial and share similar phase properties. 
In this section, we propose the problem of matrix simultaneous alignment and define a measure of matrix diversity, which is foundational for addressing the synchronization problem. 
	
Let $\mathcal A=\{A_i \in\mathbb{C}^{m\times m}: i=1,\dots,n\}$. Consider $\mathrm{rank}(A_i)$ as the vitality of $A_i$.

\begin{definition}
Given $\alpha \in [0,\frac{\pi}{2})$, the set $\mathcal{A}$ is said to be simultaneously $\alpha$-alignable if there exists a $K\in\mathbb{C}^{m\times m}$ such that $\mathrm{rank}(A_iK)=\mathrm{rank}(A_i)$ and
		$
		\Phi(A_iK)\subset [-\alpha,\alpha]
		$
		for $i=1,\dots, n$.
	\end{definition}

Note that the rank condition is needed to make the definition meaningful. Without it $K=0$ would be able to set any set of matrices into zero matrices, perfectly aligning them in phase but destroying all the vitality. Obviously, if $\mathcal{A}$ is $\alpha$-alignable for some $\alpha \in [0, \frac{\pi}{2})$, then it is $\beta$-alignable for all $\beta \in (\alpha, \frac{\pi}{2})$. 

The following lemma is useful in verifying the feasibility of simultaneous alignment.
\begin{lemma}\label{lem: singlematrixlmi}
Let $A\in\mathbb{C}^{m\times m}$ and $\alpha\in[0,\frac{\pi}{2})$. There exists a $K\in\mathbb{C}^{m\times m}$ such that $\Phi(AK)\subset [-\alpha,\alpha]$ and $\mathrm{rank}(AK)=\mathrm{rank}(A)$ if and only if there exists a $K\in\mathbb{C}^{m\times m}$ such that $\mathrm{Re}(AK)\geq 0$, $\mathcal R(AK)=\mathcal R(A)$ and
\begin{align*}
-\tan\alpha\ \mathrm{Re} (AK)\leq\mathrm{Im} (AK)\leq\tan\alpha\ \mathrm{Re} (AK).
\end{align*}
\end{lemma}
\begin{proof}
In view of Lemma~\ref{lem: phaselmi}, it suffices to show the equivalence between $\mathcal R(AK)=\mathcal R(A)$ and $\mathrm{rank}(AK)=\mathrm{rank}(A)$.
 If $\mathcal R(AK)=\mathcal R(A)$, then $\mathrm{rank}(AK)=\mathrm{rank}(A)$. Since $\mathcal R(AK)\subset\mathcal R(A)$, then $\mathrm{rank}(AK)=\mathrm{rank}(A)$ implies $\mathcal R(AK)=\mathcal R(A)$. The proof is completed.
\end{proof}
The simultaneous alignment can be checked and the aligning $K$ can be determined by solving a set of LMIs.
\begin{proposition}\label{prop:LMI}
$\mathcal A$ is simultaneously $\alpha$-alignable if and only if the following LMIs are feasible
\begin{align*}
\mathrm{Re}(A_iK)&\geq A_iA_i^*,\\
-\tan\alpha\ \mathrm{Re} (A_iK)&\leq\mathrm{Im} (A_iK)\leq\tan\alpha\ \mathrm{Re} (A_iK),
\end{align*}
for $i=1,\dots,n$.
\end{proposition}
\begin{proof}
In view of Lemma~\ref{lem: singlematrixlmi}, it suffices to show $\mathcal R(A_iK)=\mathcal R(A_i)$ if and only if $\mathrm{Re}(A_iK)\geq A_iA_i^*$.
Since $\mathcal{R}(A_iK)\subset\mathcal R(A_i)$, it is equivalent to establish the condition that ensures $\mathcal{R}(A_i)\subset\mathcal R(A_iK)$. Note that $A_iK$ is quasi-strictly accretive, hence $\mathcal R(\mathrm{Re}(A_iK))=\mathcal R(A_iK)$. Moreover, it holds that $\mathcal R(A_i)=\mathcal R(A_iA_i^*)$. Both $\mathrm{Re}(A_iK)$ and $A_iA_i^*$ are positive semidefinite. The existence of a $K$ such that $\mathcal R(A_iK)\supset\mathcal{R}(A_i)$ is equivalent to the existence of a $K$ satisfying $\mathrm{Re}(A_iK)\geq A_iA_i^*$. The proof is completed.
\end{proof}
The aligning $K$ is not unique, for example, if $K$ is a solution, then any scaled version $\mu K$ with $\mu>1$ is also a solution.
By exploiting the simultaneous alignment, we have the following definition.
\begin{definition}
For a matrix set $\mathcal{A}$, let 
\[
\alpha(\mathcal{A}) = \big\{\alpha \in [0,\frac{\pi}{2}) : \text{$\mathcal{A}$ is simultaneously $\alpha$-alignable.} \big\} .
\]
The diversity of $\mathcal A$ is defined as
\[
\mathrm{div}(\mathcal{A}) = \begin{cases}
    \inf \alpha(\mathcal{A}) & \text{if $\alpha(\mathcal{A})$ is nonempty}, \\
    \frac{\pi}{2} & \text{otherwise.}
\end{cases}
\]    
\end{definition}

An interesting feature of this definition is that the diversity is not directly defined as the difference among the individuals. Rather, it is defined as how similar they can be made after some simultaneous uniform operation. The particular operation in this definition is right multiplication by $K$. It is easy to see the diversity has the following properties:
\begin{enumerate}
    \item If $n=1$, then $\mathrm{div} (\mathcal{A})=0$.
    \item If $A_i$ are all the same, then $\mathrm{div} (\mathcal{A}) = 0$.
    \item If $A_i \!=\mu_i A$ with $\mu_i\!>\!0, i\!=\!1,\dots,n$, then $\mathrm{div} (\mathcal{A})\!=\!0$.
    \item If $A_i$ are all positive semi-definite, then $\mathrm{div} (\mathcal{A}) = 0$.
    \item For nonzero matrix $A$, $\mathrm{div} \{A,-A\}=
        \pi/2.$
    \item If $\tilde{\mathcal{A}}\subset \mathcal{A}$, then $\mathrm{div} (\tilde{\mathcal{A}})\leq \mathrm{div} (\mathcal{A})$.
\end{enumerate}
One can easily device a bisection algorithm to approximately compute $\mathrm{div} \mathcal{A}$ with arbitrary precision. Initially we know that the diversity is in an interval $[0, \pi/2]$. By checking whether $\mathcal{A}$ is $\pi/4$-simultaneously alignable using Proposition \ref{prop:LMI}, one can determine whether $\mathrm{div} \mathcal{A}$ is in either of the two half sub-intervals. Continuing to halve the intervals this way, $\mathrm{div} \mathcal{A}$ can be approximated with an arbitrarily small error.

Too large a diversity of $\mathcal{A}$ may make the set difficult to handle. This observation motivates the matrix clustering problem: whether we can partition the set into a small number of clusters so that each cluster has a sufficiently small diversity. If this is possible, we may then deal with the clusters one by one with ease. The clustering problem does not have an easy solution. It usually falls into the type of problems that the machine learning community is interested in. The development of efficient matrix clustering algorithms is an important future research topic.

\subsection{Network symmetry}
A matrix may not be semi-sectorial but can be made semi-sectorial through diagonal similarity transformation. Laplacian matrices fall into this class. The essential phases have been defined for such a class of matrices in \cite{wd2023}. The essential phase of a Laplacian matrix is particularly useful in the study of networks. It provides a measure of the network symmetry.

Given $A\in\mathbb{C}^{m\times m}$, its largest and smallest essential phases are given by
\begin{align*}
    \overline\phi_{\mathrm{ess}}(A)=\inf_{D\in\mathcal{D}} \overline\phi(D^{-1}AD)
    \ \text{and}\
    \underline\phi_{\mathrm{ess}}(A)\!=\!\sup_{D\in\mathcal{D}} \underline \phi(D^{-1}AD),
\end{align*}
where $\mathcal{D}$ is the set of positive definite diagonal matrices such that $D^{-1}AD$ is semi-sectorial. 
If $A$ is a real matrix, then $\underline\phi_{\mathrm{ess}}(A)=-\overline\phi_{\mathrm{ess}}(A)$. In this case, we denote $\overline\phi_{\mathrm{ess}}(A)$ by $\phi_{\mathrm{ess}}(A)$ for notational simplicity.


We first review the essential phases of Laplacians corresponding to strongly connected graphs. The Laplacian is semi-sectorial if and only if the digraph is balanced \cite{wd2023}. For an unbalanced graph, an analytic expression for the essential phase of the Laplacian can be obtained. To be specific, let $v$ be a positive left eigenvector of $L$ corresponding to zero eigenvalue. Let $V=\mathrm{diag}\{v\}$ and $D_0=V^{-1/2}$. 
	\begin{lemma}[\cite{wd2023}]\label{lem: essentiallaplacian}
		For a Laplacian matrix $L$ corresponding to a strongly connected graph, it holds that 
		\begin{align*}
		\phi_{\mathrm{ess}}(L)\!=\!\overline\phi(D_0^{-1}LD_0).
		\end{align*}
		
	\end{lemma}
\

The essential phase provides a measure of the network symmetry. One can see that $VL$ is a Laplacian matrix with $\mathbf{1}_n$ being a common left and right eigenvector corresponding to eigenvalue zero. This means that $VL$ is the Laplacian matrix of a balanced graph. Further, if $VL$ is symmetric, we say the graph
corresponding to $L$ is essentially undirected.
\begin{lemma}
There holds $\phi_{\ess}(L)=0$ if and only if the digraph is essentially undirected.   
\end{lemma}

The essential phase of a Laplacian corresponding to a graph that is not strongly connected but has a spanning tree has been studied in \cite{wd2023} as well. The Laplacian matrix in this case is reducible, and can be written in the Frobenius normal form.
\begin{lemma}[\cite{BrualdiRyser}] 
If the graph has a spanning tree, then under a proper permutation the Laplacian matrix can be written as a block lower triangular matrix
\begin{align}\label{eq: laplacian}
    L=\begin{bmatrix}
        L_{11} & 0&\dots &0 \\
        L_{21}&L_{22} &\dots&0\\
        \vdots&\vdots&\ddots&\vdots\\
        L_{\mu 1} &L_{\mu 2}&\dots&L_{\mu\mu}
    \end{bmatrix},
\end{align}
where $L_{11}$ is an irreducible Laplacian matrix or a zero matrix with dimension one and $L_{jj}, j=2,\dots,\mu$, is irreducible with at least one row having positive row sum.
\end{lemma}

The induced subgraph of $L_{jj}, j=1,\dots, \mu$ corresponds to a strongly connected component of the graph. We will call it a component in later development of this paper. 
The Laplacian $L$ can not be made semi-sectorial through diagonal similarity transformation with its left Frobenius eigenvector in this case. Notwithstanding, the essential phases of $L_{jj}, j=1,\dots,\mu,$ exist and take effect in the synchronization context. Note that $L_{11}$ is a Laplacian of a strongly connected component formed by all the roots, whose essential phase is given by Lemma \ref{lem: essentiallaplacian}. The matrices $L_{jj}, j=2,\dots,\mu$ are nonsingular M-matrices. The upper bounds of their essential phases are given by the following lemma. 
	
\begin{lemma}
    For a Laplacian matrix $L$ in the form (\ref{eq: laplacian}), there holds 
    \[
    \phi_{\mathrm{ess}}(L_{jj})\!\leq\!\overline\phi(D_{j}^{-1}L_{jj}D_{j})\leq \frac{\pi}{2}, j=2,\dots, \mu,
    \] 
    where $D_{j}=\mathrm{diag}(\sqrt{x_{j1}/y_{j1}},\dots,\sqrt{x_{jn}/y_{jn}})$, $x_j$ and $y_j$ are right and left eigenvector corresponding to the smallest real eigenvalue of $L_{jj}$.
\end{lemma}
	\begin{proof}
		According to the Gershgorin theorem, the spectrum of $L_{jj}$ lies in the right half plane. Furthermore, the matrix $L_{jj}$ is real. Hence if $D^{-1}L_{jj}D$ is semi-sectorial for some $D$, it must also be accretive, which implies $\phi_{\mathrm{ess}}(L_{jj})\in[0,\pi/2]$. The next step is to show that $\overline\phi(D_{j}^{-1}L_{jj}D_{j})$ provides a less conservative upper bound for $\phi_{\mathrm{ess}}(L_{jj})$.
		
		Since $L_{jj}$ is an M-matrix, it can be written as $s_jI-A_j$, where $A_j$ is nonnegative and $s_j>\rho(A_j)$. It can be seen that $D_{j}^{-1}A_jD_{j}$ has a common left and right eigenvector corresponding to $\rho(A_j)$, hence $\rho(A_j)$ is a sharp point of $W(A_j)$ according to \cite{li2002numerical}. It follows that $s_j-\rho(A_j)$ is a sharp point of $W(L_{jj})$ since $W(L_{jj})$ is a translation of $-W(A_{j})$. As a consequence, $\overline\phi(D_{j}^{-1}MD_{j})$ provides an upper bound of $\phi_{\ess}(L_{jj})$.
	\end{proof}
	The essential phases of $L_{jj}$ in general do not have a closed form expression. Notwithstanding, an algorithm for the numerical computation has been derived. Details can be found in \cite{wd2023}.

	\section{Synchronization controller design}\label{sec:synthesis}
	In this section, we delve into the problem of synchronization synthesis, with the goal of designing the component controller and component-cluster combined controller architecture that enforce synchronization among heterogeneous agents.
 
 This problem is particularly challenging due to two major considerations. The first is the synchronizability of the network, i.e., whether there exist controllers capable of achieving synchronization across diverse agents. It will be delineated that the phase serves as a key characterization of the diversity of the agents. The second challenge is to provide a construction method of controllers for the synchronizable multi-agent system. 

 By properly labeling the agents, the Laplacian matrix can be written in the form (\ref{eq: laplacian}). Divide the agents into $\mu$ components according to (\ref{eq: laplacian}), denoted by $\mathcal P_1,\dots, \mathcal P_\mu$, each representing a strongly connected component. The size of each component corresponds to the dimension of $L_{jj}$. Notably, the first component contains all the roots of the graph, serving as the steering component. These agents have directed paths to all other nodes but do not receive information from other components. When $L_{11}$ is a zero matrix of dimension one, the multi-agent system has a single leader. The agents in $j$-th component for $j=2,\dots,\mu$ can be treated as followers. 
	
\subsection{Synchronization under component-wise controllers}
For large-scale networks, scalability becomes a critical concern. In this context, designing controllers based on the graph topology offers a practical and efficient approach. Specifically, we consider component-wise controllers design, where the number of controllers is equal to the number of strongly connected components. Denote the component-wise controllers by $C_j(z), j=1,\dots,\mu$.

Synchronizing a large group of heterogeneous agents using a component-wise controller framework is deeply connected to the simultaneous stabilization problem, a well-known challenge in control theory that has been extensively studied in works such as \cite{vidyasagar1985control, saeks1982fractional, cao1999simultaneous}. However, the systems in our study exhibit a unique structure: they are semi-stable and share common poles located on the unit circle. This distinctive property simplifies the problem, enabling a more efficient solution approach compared to the general case. Let 
 $
\mathcal{N}_{lj}=\{N_{li}: \bm P_i\in\mathcal P_j\}$, where $l=0,1,\dots,q,\pi$, $i=1,\dots, n$, $j=1,\dots,\mu.
 $
We have the following result. The proof is given in the Appendix. 
\begin{theorem}\label{thm:component}
Problem \ref{problem1} is solvable with component-wise controllers for each agent if 
 \[
\mathrm{div}(\mathcal{N}_{lj})+\phi_{\mathrm{ess}}(L_{jj})<\frac{\pi}{2},
\]
for $l=0,1,\dots,q,\pi$, and $j=1,\dots,\mu$.
\end{theorem}
Denote the aligning matrix in $\mathcal{N}_{lj}$ by $K_{lj}.$ The component-wise controllers are given by $C_j(z)=\epsilon K_j(z), j=1,\dots,\mu$, for all $\epsilon\in(0,\epsilon^*)$, where $K_j(z)\in\mathcal{RH}_\infty^{m\times m}$ satisfies the interpolating conditions $K_j(e^{j\omega_l})=e^{j\omega_l}K_{lj}$ and $\epsilon^*>0$ can be estimated from given data. 

	\begin{remark}
		When $L_{11}$ is a zero matrix with dimension one, indicating that the agents have a single leader, the conditions in Theorem~\ref{thm:component} only need to be satisfied within components of followers, i.e., for $j=2,\dots,\mu$. 
	\end{remark}
The theorem shows that in a network of heterogeneous individuals, whether a good collective behaviour can be reached depends on the trade-off between the individuals' diversity and the interaction quality. 
 
	Checking whether the conditions in Theorem~\ref{thm:component} are satisfied is an LMI feasibility problem. This can be achieved by either first computing the $\mathrm{div}(\mathcal{N}_{lj})$ using a bisection method and verifying if an upper bound of $\phi_{\ess}(L_{jj})$ is obtained at each step during the bisection process, or by directly computing $\phi_{\ess}(L_{jj})$ and verifying whether the agents can be simultaneously aligned to corresponding phase interval during the process of bisection search of $\phi_\ess(L_{jj})$. After the aligning matrices $K_{0j}, K_{\pi j}\in\mathbb{R}^{m\times m}, K_{lj}\in\mathbb{C}^{m\times m},l=1,\dots,q$, are obtained, we provide a design procedure of $K_j(z)\in\mathcal{RH}_\infty^{m\times m}$ such that 
	\begin{align*}
	K_j(e^{j\omega_l})&={e^{j\omega_l}}K_{lj}, l=0,1,\dots,q,\pi,\\
		K_j(e^{-j\omega_l})&=e^{-j\omega_l}\bar K_{lj}, l=1,\dots,q.
	\end{align*}
	Let $z_0=1,$ and $z_{2l-1}=e^{j\omega_l}, z_{2l}=e^{-j\omega_l}$ for $l=1,\dots,q$, $z_{2q+1}=-1$. With the aid of Lagrange polynomial, a $K_j(z)$ can be given by
	\begin{align*}
    K_j(z)&\!=\!\sum_{l=1}^{q}\left(\!K_{lj}\!\prod_{i=0,\atop i\neq 2l-1}^{2q+1}\frac{z^{-1}-z_i^{-1}}{z_{2l-1}^{-1}-z_i^{-1}}\!+\!\bar K_{lj}\prod_{i=0,\atop i\neq 2l}^{2q+1}\frac{z^{-1}-z_i^{-1}}{z_{2l}^{-1}-z_i^{-1}}\right)\\
&\!+\!K_{0j}\prod_{i=1}^{2q+1}\frac{z^{-1}-z_i^{-1}}{1-z_i^{-1}}+K_{\pi j}\prod_{i=0}^{2q}\frac{z^{-1}-z_i^{-1}}{-1-z_i^{-1}}.
	\end{align*}
	
	From the proof of Theorem~\ref{thm:component}, the synchronized output value is only determined by the initial state of agents in the first component. Other agents have no contributions to it. The problem can be separated into two parts: the synchronization in the steering component and tracking of other agents. When the graph is strongly connected, the number of strongly connected components equals $1$. In this case, a uniform controller can be applied, and Theorem~\ref{thm:component} simplifies accordingly.
	\begin{corollary}\label{cor: strongly connected}
	Problem \ref{problem1} is solvable with a uniform controller under a strongly connected directed graph if 
  \[
\mathrm{div}(\mathcal{N}_{l1})+\phi_{\ess}(L)<\frac{\pi}{2}
\]
for $l=0,1,\dots,q,\pi$. 
\end{corollary}
 If the graph is undirected, the essential phases of the corresponding Laplacian matrix is $0$. 
\begin{corollary}\label{cor: undirected}
	Problem \ref{problem1} is solvable with a uniform controller under a connected undirected graph if 
  \[
\mathrm{div}(\mathcal{N}_{l1})<\frac{\pi}{2}
\]
for $l=0,1,\dots,q,\pi$. 
\end{corollary}
One can see that conditions in Theorem~\ref{thm:component} only require the residue information of semi-stable modes. The stable part of each agent does not appear in the theorem, showing that the proposed controller design technique can tolerate large heterogeneity among the agents. The method is robust against the perturbations. Furthermore, the synchronizability condition only depends on the phase information. The gain of each agent can be arbitrarily large. Therefore, the synchronization problem is likely solvable if the agents have vastly different sizes but similar shapes. The design of synchronizing controllers also suggests the use of low gain controller, indicating that the coordination among the agents does not need strong action. Instead
	it is more critical to have the right directions of the action.

	The consensus problem is a special case of synchronization problem, where all the agents share only one common pole, i.e.,
	\[
\bar P_i(z)=\frac{N_{0i}}{z-1}.
    \]
	\begin{corollary}
		The multi-agent system is consensusable with component-wise controllers if
\[
\mathrm{div}(\mathcal{N}_{0j})+\phi_{\mathrm{ess}}(L_{jj})<\frac{\pi}{2}
\]
for $j=1,\dots,\mu.$ 
	\end{corollary}
 
	The very early studies assume that the agents are simply all identical integrators $x_i(t+1)=x_i(t)+u_i(t)$. While the integrator in continuous-time is passive, the integrator in discrete-time is not passive resulting from the sampling process. The controllers $C_i(z)=K, i=1,\dots,n,$ with a sufficient small positive $K$ solves the problem, which is consistent with the result in the literature. 
	
      While the primary focus of this subsection has been on the design and advantages of component-wise controllers, it is also worth considering scenarios where a uniform controller might be applicable. The uniform controller, though less flexible than component-wise designs, offers simplicity by assigning the same controller to all agents. Here we briefly discuss the conditions under which synchronization can be achieved using a uniform controller. Let 
 $
\mathcal{N}_{l}=\{N_{lj}:j=1,\dots,n\}.
 $ 
\begin{theorem}\label{thm: uniform}
Problem \ref{problem1} is solvable with a uniform controller if 
 \[
\mathrm{div}(\mathcal{N}_{l})+\max_{1\leq i\leq \kappa}\{\phi_{\mathrm{ess}}(L_{ii})\}<\frac{\pi}{2},
\]
for $l=0,1,\dots,q,\pi$.
 \end{theorem}
 The proof follows a similar approach to Theorem~\ref{thm:component} and is therefore omitted for brevity.
 Denote the aligning matrix in $\mathcal{N}_{l}$ by $K_{l}.$ 
 A uniform controllers is given by $C(z)=\epsilon K(z)$, for all $\epsilon\in(0,\epsilon^*)$, where $K(z)\in\mathcal{RH}_\infty^{m\times m}$ satisfies the interpolating conditions $K(e^{j\omega_l})=e^{j\omega_l}K_{l}$ and $\epsilon^*>0$ can be estimated from given data. The condition here is more conservative than that in Theorem~\ref{thm:component}. However, the aligning matrix $K_l$ in this case can synchronize all agents, whereas in Theorem~\ref{thm:component}, the aligning matrices depend on the individual components.

 \subsection{Synchronization under component-cluster combined controller architecture}
Although the component-wise controller scheme offers improved scalability compared to the agent-specific controller scheme, it may still impose constraints that make the synchronization problem unsolvable if the diversity of any of the agent components is too large. If such limitations arise, adopting a fully agent-specific controller can solve the problem but sacrifices scalability. Therefore, we aim at identifying a solution that balances conservatism and scalability effectively.

To address this challenge, we analyze the limitations of the component-wise controller. As stated in Theorem~\ref{thm:component}, achieving synchronization with this controller structure requires that the diversity of the residue matrices of all agents at each persistent mode and in each network component remains within a specific range. However, as the diversity of these matrices increases, meeting this condition becomes increasingly difficult. Consequently, greater variation in the residue matrices significantly complicates the controller design.

This observation motivates a divide and conquer strategy. By grouping agents with similar physical properties into clusters, we expect that the diversity of each cluster is relatively small and then we tailor the componentwise controller design to each cluster, ensuring a more flexible and scalable solution. Specifically, a common controller can be assigned to agents in each intersection of an agent cluster and an agent component. This approach reduces conservatism by distributing design complexity across clusters.

To formalize this approach, we assume the clustering is done based on physical understanding of the agents or based on a machine-learning process from the data, such as clustering in UAV networks \cite{luo2024star} or using saddle-point analysis \cite{burger2012hierarchical}. These methods generate clusters that align with the system’s characteristics. Denote the set of agents by $\mathcal P$. From the network strongly connected component structure that was used earlier, $\mathcal P$ is partitioned into $\mu$ components, $\{\mathcal{P}_1,\dots,\mathcal{P}_\mu\}$, each forming a strongly connected component. Now $\mathcal P$ is also partitioned into $\nu$ clusters, denoted by $\{\mathcal G_1, \dots, \mathcal G_\nu\}$. Each agent belongs to both a component $\mathcal P_j$ and a cluster $\mathcal G_k$. Let $\mathcal S_{jk} =\mathcal{P}_j\cap \mathcal{G}_k,  j=1,\dots,\mu, k=1,\dots,\nu$. Then $\{S_{jk}: j=1,\dots,\mu, k=1,\dots,\nu\}$ is also a partition of $\mathcal P$. We intend to assign a common controller to all agents in each $\mathcal{S}_{jk}$. Recall that the set of partitions forms a partially ordered set under the finer/coarser partial order. This partially ordered set happens to be a lattice \cite{mac2023algebra}. In this language, the partition $\{\mathcal{S}_{jk}: j=1,\dots,\mu, k=1,\dots,\nu\}$ is exactly the meet of the partitions $\{\mathcal{P}_j: j=1,\dots,\mu\}$ and $\{\mathcal{G}_k: k=1,\dots,\nu\}$. We assume that each $\mathcal S_{jk}$ is equipped with a common controller. Such a controller architecture is called a component-cluster combined controller (4C) architecture. Define
$
\mathcal N_{ljk}=\{N_{li}: P_i\in\mathcal S_{jk}\}.
$
\begin{theorem}\label{thm: cluster}
    Problem (\ref{problem1}) is solvable under 4C architecture if 
    \[
    \mathrm{div}(\mathcal N_{ljk})+\phi_\ess(L_{jj})<\frac{\pi}{2},
    \]
    for $l=0,1,\dots,q,\pi$,  $j=1,\dots,\mu$, $k=1,\dots,\nu$.
\end{theorem}
The proof follows a similar approach to Theorem~\ref{thm:component} and is therefore omitted for brevity. Denote the aligning matrix in $\mathcal{N}_{ljk}$ by $K_{ljk}.$ The component-cluster combined controllers are given by $C_{jk}(z)=\epsilon K_{jk}(z), j=1,\dots,\mu, k=1,\dots,\mu$, for all $\epsilon\in(0,\epsilon^*)$, where $K_{jk}(z)\in\mathcal{RH}_\infty^{m\times m}$ satisfies the interpolating conditions $K_{jk}(e^{j\omega_l})=e^{j\omega_l}K_{ljk}$ and $\epsilon^*>0$ can be estimated from given data. Theorem~\ref{thm: cluster}
provides the solvability condition for synchronization synthesis problem with component-cluster combined controllers. Verifying the conditions in Theorem~\ref{thm: cluster} amounts to solving an LMI feasibility problem. Once the LMIs are satisfied, the synchornizing controller for each component-cluster intersection can be designed analogously to the component-wise controller.

If all agents share only one common persistent pole at $1$, Theorem~\ref{thm: cluster} simplifies to the following corollary.
\begin{corollary}
  The multi-agent system is consensusable under 4C architecture if 
    \[
    \mathrm{div}(\mathcal N_{jk})+\phi_\ess(L_{jj})<\frac{\pi}{2}, 
    \]
    for $j=1,\dots,\mu$, $k=1,\dots,\nu$.
\end{corollary}



	\section{Simulation}\label{sec:simulation}
In this section, we use a numerical example with five agents to illustrate our theoretical result in Theorem~\ref{thm:component}.
	
		Consider a group of agents $\bm P_1,\dots,\bm P_5$ with network shown in Fig. \ref{fig:directedgraph}.
		\begin{figure}[htbp]
			\begin{center}
				\includegraphics[width=3cm]{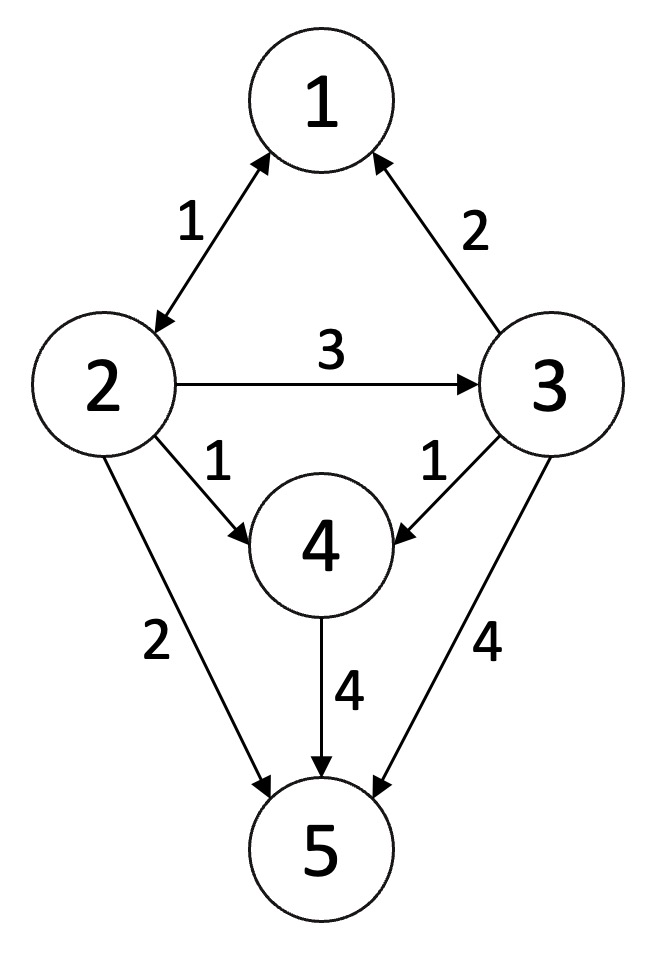}    
				\caption{\label{fig:directedgraph}A directed graph.}
			\end{center}
		\end{figure}
		Assume the agent dynamics are given by 
		\begin{align*}
			\bm P_i&=\bar {\bm P}_i+\bm\Delta_i, i=1,\dots,5,
		\end{align*}
		where the transfer functions of $\bar {\bm P}_i,i=1,\dots,5$ are given by
		\begin{align*}
			\bar P_1(z)=&\frac{\begin{bmatrix} 3.4 &2.8\\1.1 &-0.3\end{bmatrix}}{z-1}+\frac{\begin{bmatrix} -8z-23.5 &  -3.8z-7.1\\
					-21z-4.8  &-10z-8.5\end{bmatrix}}{z^2-\sqrt 2z+1},
		\end{align*}
		\begin{align*}
			\bar P_2(z)=&\frac{\begin{bmatrix}0.8 & 6.6\\2.2 & -3.6\end{bmatrix}}{z-1}+\frac{\begin{bmatrix} 4z-12.7 &  12z-13.6\\
					22z-8.9  & 8.2z-4.7\end{bmatrix}}{z^2-\sqrt 2z+1},
		\end{align*}
		\begin{align*}
			\bar P_3(z)=&\frac{\begin{bmatrix}3.4 & 3.8\\0.9 & 0.3\end{bmatrix}}{z-1}+\frac{\begin{bmatrix} -2.2z+10.1 &  18z-16.3\\40z-23.1  & -1.6z-11.1\end{bmatrix}}{z^2-\sqrt 2z+1},
		\end{align*}
		\begin{align*}
			\bar P_4(z)=&\frac{\begin{bmatrix}1.5 & 1.5\\0.1 & -4.3\end{bmatrix}}{z-1}+\frac{\begin{bmatrix} -36.2z + 3.6 &  -22.5 z + 11.8\\-6z - 8.9 & -0.3 z - 12.4\end{bmatrix}}{z^2-\sqrt 2z+1},
		\end{align*}
		\begin{align*}
			\bar P_5(z)=&\frac{\begin{bmatrix}1.5 & 1.8\\1 &-4.3\end{bmatrix}}{z-1}+\frac{\begin{bmatrix} 7.5z-18.7 &  8.7z-4.4\\50.2z-67.0  & 8.7z-29.2\end{bmatrix}}{z^2-\sqrt 2z+1},
		\end{align*}
	and each $\bm\Delta_i$ for $i=1,\dots,5$, represents a nonlinear stable dynamic term. These terms model a variety of nonlinear behaviors, including saturation, dead zone, limit cycle oscillations, nonlinear damping, and logistic map dynamics.
 
		The network topology is a directed graph which has a spanning tree. The corresponding Laplacian matrix is given by 
		\begin{align*}
			L= \begin{bmatrix}
				3 & -1& -2 & 0 & 0\\
				-1 & 1 & 0 & 0 & 0\\
				0 & -3 & 3 & 0 & 0\\
				0 & -1 & -1 & 2 & 0\\
				0 & -2 & -4 & -4 & 10
			\end{bmatrix}.
		\end{align*}
		Here 
		\[
		L_{11}=\begin{bmatrix}3 & -1& -2 \\
			-1 & 1 & 0 \\
			0 & -3 & 3 \end{bmatrix}, L_{22}=2, L_{33}=10.
		\] It follows that $\phi_{\mathrm{ess}}(L_{11})=0.3614$, $\phi_{\mathrm{ess}}(L_{22})=0$ and $\phi_{\mathrm{ess}}(L_{33})=0$. The agents can be divided into three strongly connected components accordingly. 
		
		Checking the LMIs condition yields
		\[
K_0=\begin{bmatrix}-1.8& 7.0\\
    5.3& -3.4\end{bmatrix},K_1=\begin{bmatrix} -2.6 - 5.5i & 1.1 + 3.5i\\
   7.1 + 1.4i & -0.1-11.0i\end{bmatrix}.
		\]
		Here, a uniform controller can be obtained, which is given by \[
\begin{bmatrix} \displaystyle \frac{-12.9z^2+21.1z-9.2}{z^2}  &\displaystyle \frac{13.7z^2-21.7z+12.2}{z^2}\\
  \displaystyle \frac{12.7z^2-12.3z+2.7}{z^2} & \displaystyle \frac{-22.2z^2+42.2 z-22.1z}{z^2}
  \end{bmatrix}.
  \] This controller can indeed enforce synchronization as confirmed in Fig~\ref{fig:designtrajectory1} and Fig~\ref{fig:designtrajectory2}.

	\begin{figure*}[htbp]
		\begin{center}
			\includegraphics[width=11cm]{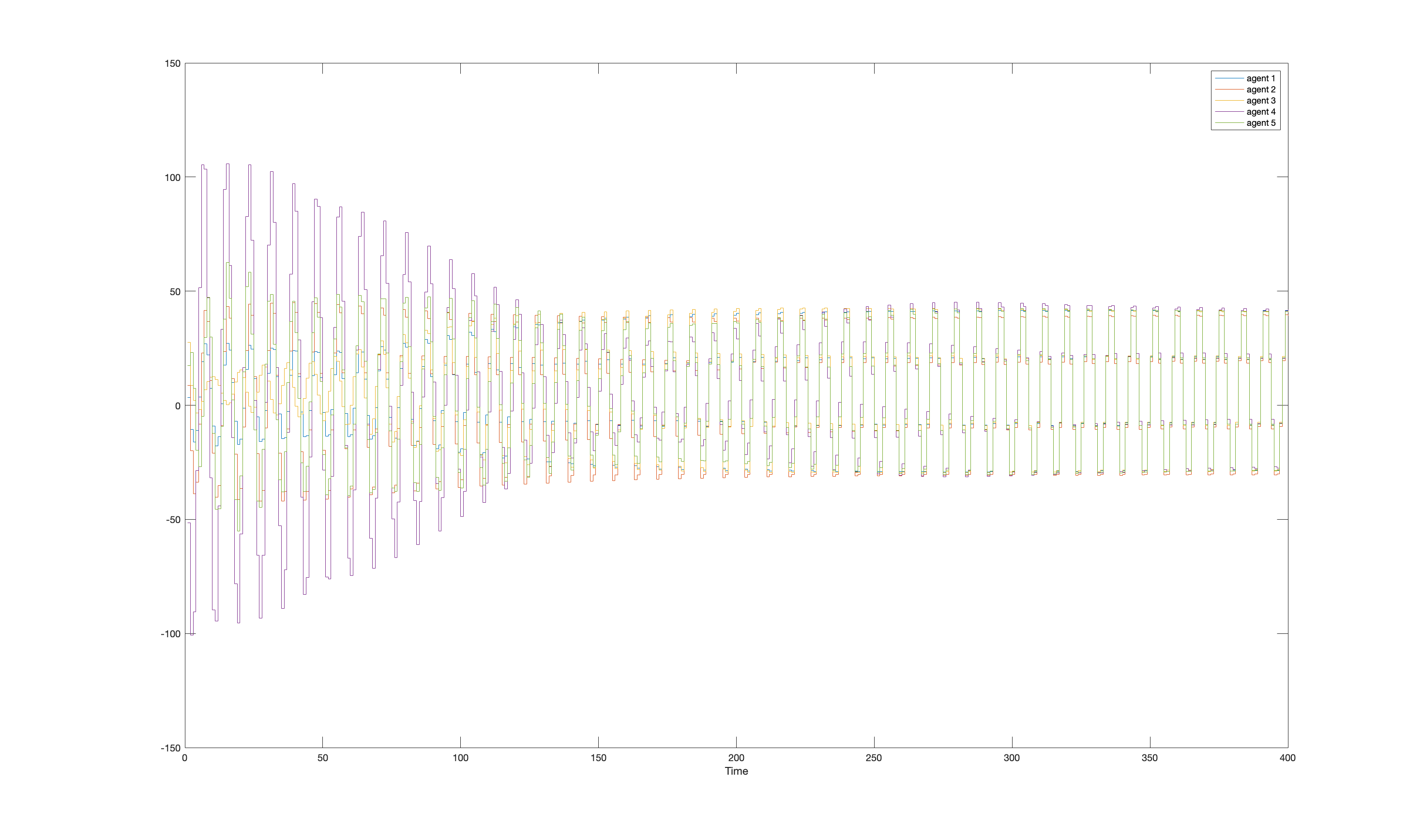}    
			\caption{\label{fig:designtrajectory1}Trajectories of first output.}
		\end{center}
		\begin{center}
			\includegraphics[width=11cm]{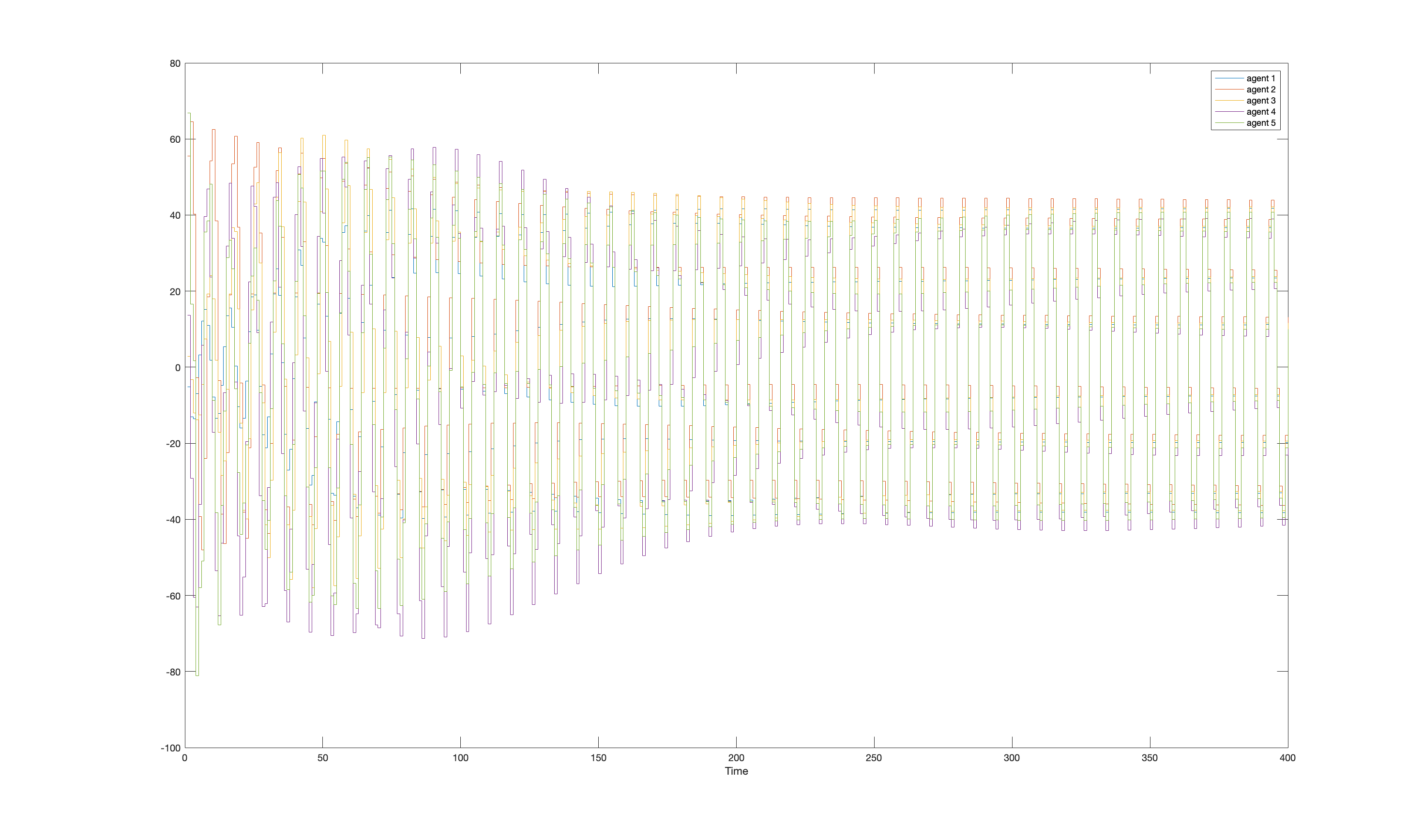}    
			\caption{\label{fig:designtrajectory2}Trajectories of second output.}
		\end{center}
	\end{figure*}

\section{Conclusion}\label{sec:conclusion}
In this paper, we investigated synchronizing controller architecture and design for a heterogeneous discrete-time nonlinear dynamical network from a novel phase-based perspective. We characterized the maximum diversity among agents that can be accommodated by the network topology and obtained solvability conditions based on this diversity. Additionally, we developed design algorithms for synchronizing controllers under component-wise and 4C architecture scenarios. Through these results, we aim to highlight the unique insights that matrix phase offers in uncovering synchronization properties that remain hidden in traditional approaches.

A technical challenge that we face is how tight the condition in Theorem~\ref{thm:component} is. More specifically, is the condition in any sense necessary? We believe that the answer is yes. The investigation towards answering this question precisely is intensively on going.

In future work, we aim to develop clustering algorithms with diversity requirements to improve network synchronizability and control effectiveness. Additionally, we plan to explore the graph design problem, focusing on how network topology can be structured or modified to ensure robust synchronization. This includes examining conditions under which specific network configurations promote stability and scalability, as well as understanding how structural changes impact overall synchronizability.

 \section*{References}
 \bibliographystyle{IEEEtran}
\bibliography{network}
	
\appendices
\section{Useful lemmas}
\begin{lemma}[\cite{meng2010distributed}]\label{lem: laplaciancomb} 
		For the partitioned Laplacian matrix in (9), it holds that $-L_{jj}^{-1}\begin{bmatrix}L_{j1} &\dots&L_{j(j-1)}\end{bmatrix}$ is nonnegative, and each of its rows has unity sum for $j=2,\dots,\mu$.
	\end{lemma}
\begin{lemma}[\cite{rantzer1996kalman}]\label{lem:brl}
		Let $G(z)={\mbox{$\left[
				\begin{array}{c|c}
					A & B \\ \hline \\[-4mm]
					C & D
				\end{array}
				\right]$}}$. Assume that $A$ is stable. Then $\|G\|_\infty<\gamma$ if and only if there exists a matrix $X>0$ such that 
		\[
		\begin{bmatrix}
			A^*XA-X+C^*C & A^*XB+C^*D\\
			B^*XA+D^*C & B^*XB-\gamma^2 I+D^*D
		\end{bmatrix}<0.\]
	\end{lemma}
 
\section{Proof of Theorem~\ref{thm:component}}
\begin{proof}
If the condition
\[
\mathrm{div}(\mathcal{N}_{lj})+\phi_{\mathrm{ess}}(L_{jj})<\frac{\pi}{2}
\]
holds for $l=0,1,\dots,q,\pi, j=1,\dots,\mu$, then aligning matrices $K_{lj}$ can be obtained so that 
\begin{align}\label{eq:phasebound}
\Phi(N_{li}K_{lj}) \subset \left[-\frac{\pi}{2}+\phi_{\mathrm{ess}}(L_{jj}),\frac{\pi}{2}-\phi_{\mathrm{ess}}(L_{jj})\right]
\end{align}
for all $i$ with $\bm P_i\in\mathcal P_j$.
Moreover, a $K_j(z)\in\mathcal{RH}_\infty^{m\times m}$ satisfying the interpolating conditions 
\begin{align*}
K_j(e^{j0})&=K_{0j},\\
K_j(e^{j\pi})&=-K_{\pi j},\\
K_j(e^{j\omega_l})&=e^{j\omega_l}K_{lj},\\
K_j(e^{-j\omega_l})&=e^{-j\omega_l}\bar{K}_{lj}, l=1,\dots,q
\end{align*}
can be constructed by exploiting the Lagrange polynomial.
We will show that there exists an $\epsilon^*>0$ such that the multi-agent system with component-wise controllers, given by $C_j(z)=\epsilon K_j(z),\epsilon\in(0,\epsilon^*)$, can achieve synchronization. 

Let $r_j$ be the number of agents in the $j$-th components. Without loss of generality, assume the agents are ordered and grouped in the same manner as in $L_{jj}$, denoted by $\bm P_{jj}$. The corresponding signals and subsystems are denoted by $u_{jj}, y_{jj}, v_{jj}, w_{jj}, e_{jj}, \bm\Delta_{jj}, \bar{\bm P}_{jj}$ for $j=1,\dots,\mu$.
	The controllers can be grouped accordingly as follows:
	\begin{align*}
    \bm C_{jj}=\epsilon I_{r_j}\otimes \bm K_j, j=1,\dots,\mu.
	\end{align*}
        Together with \eqref{haty} and \eqref{eq: laplacian}, we have the following dynamics of each component:
		\begin{align}
			y_{11}=&(I+\epsilon\bm P_{11}(L_{11}\otimes \bm K_1))^{-1}\bm P_{11}w_{11}, \label{eq: y11}\\
			y_{jj}=&(I+\epsilon\bm P_{jj}(L_{jj}\otimes \bm K_{j}))^{-1}\bm P_{jj}   \nonumber\\ 
			&\times \!\left(w_{jj} -\epsilon(I_{r_j}\otimes \bm K_j)\sum_{i=1}^{j-1}(L_{ji}\otimes I_m)y_{ii}\right), \label{eq: yi*}
            \end{align}
		for $j=2,\dots,\mu.$ 
        
        We will first prove the synchronization of the component consisting of $\bm P_{11}$. Note that $L_{11}$ is a Laplacian matrix of a strongly connected subgraph and $y_{11}$ is independent of the other components.        
        Let $Q_1$ be an isometry matrix whose columns form a basis of $\mathrm{span}\{\mathbf{1}_{r_1}\}^\perp$ and
		\[\bm S_{11}\!=\left(I_{r_1m-m}+\epsilon(Q'_1\otimes I_m)\bm P_{11}(L_{11}Q_1\otimes \bm K_1)\right)^{-1}.\]
        As discussed in Section \ref{sec:synformulation}, the first component achieves synchronization if $\bm S_{11}$ is stable.		
		Let $v$ be a positive left eigenvector of $L_{11}$ corresponding to $0$ eigenvalue and $D_1=\mathrm{diag}\{v\}$. 
        Let $$\tilde{L}_{11}=Q_1'D_1L_{11}Q_1.$$
		Since $Q_1'Q_1\!=\!I_{r_1-1}$, $Q_1Q_1'\!=\!I_{r_1}-\frac{1}{r_1}\mathbf{1}_{r_1} \mathbf{1}_{r_1}'$, $\mathbf{1}_{r_1}'D_1L_{11}\!=\!0$, it holds that $L_{11}Q_1=D_1^{-1}Q_1Q_1'D_1L_{11}Q_1$. Thus, 
        {\small
        \begin{multline*}
        \bm S_{11}=\left( I_{(r_1-1)m}+\epsilon (Q_1'\otimes I_m)\bm P_{11}(D_1^{-1}Q_1\tilde{ L}_{11}\otimes \bm K_1)\right)^{-1},   
        \end{multline*}}
        which is the sensitivity function of the feedback system shown in Fig. \ref{fig: block diagram of S11}.          
	In order to separate the linear and nonlinear terms in $\bm P_{11}$, we apply a loop transformation shown in Fig.~\ref{fig: looptranformation}.
    Let $$\tilde{\bm \Delta}_{11}=(Q_1'\otimes I_m)\bm\Delta_{11}.$$ 
    Then $\bm S_{11}$ is stable if and only if the feedback connection of nonlinear system $\tilde{\bm \Delta}_{11}$ and LTI system $\bm M_{11}$ is stable, where
		\begin{multline*}
            \bm M_{11} =\epsilon(D_1^{-1}Q_1\tilde{L}_{11}\otimes \bm K_1) \\ \times \left(I_{(r_1-1)m} +\epsilon (Q_1'\otimes I_m) \bar{\bm P}_{11}(D_1^{-1}Q_1\tilde{L}_{11} \otimes \bm K_1)\right)^{-1}. 
		\end{multline*}
        
	\begin{figure}[htb]
			\begin{center}
            \includegraphics[width=6.5cm]{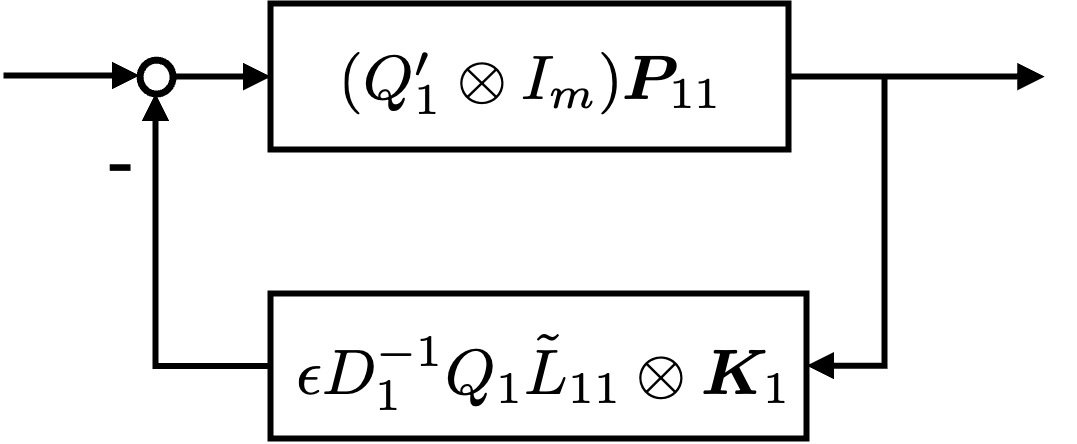}    
            \vspace{-5pt}
            \caption{\label{fig: block diagram of S11}Feedback system after diagonal scaling.}
			\end{center}
		\end{figure}	
		\begin{figure}[htb]
			\begin{center}
				\includegraphics[width=7cm]{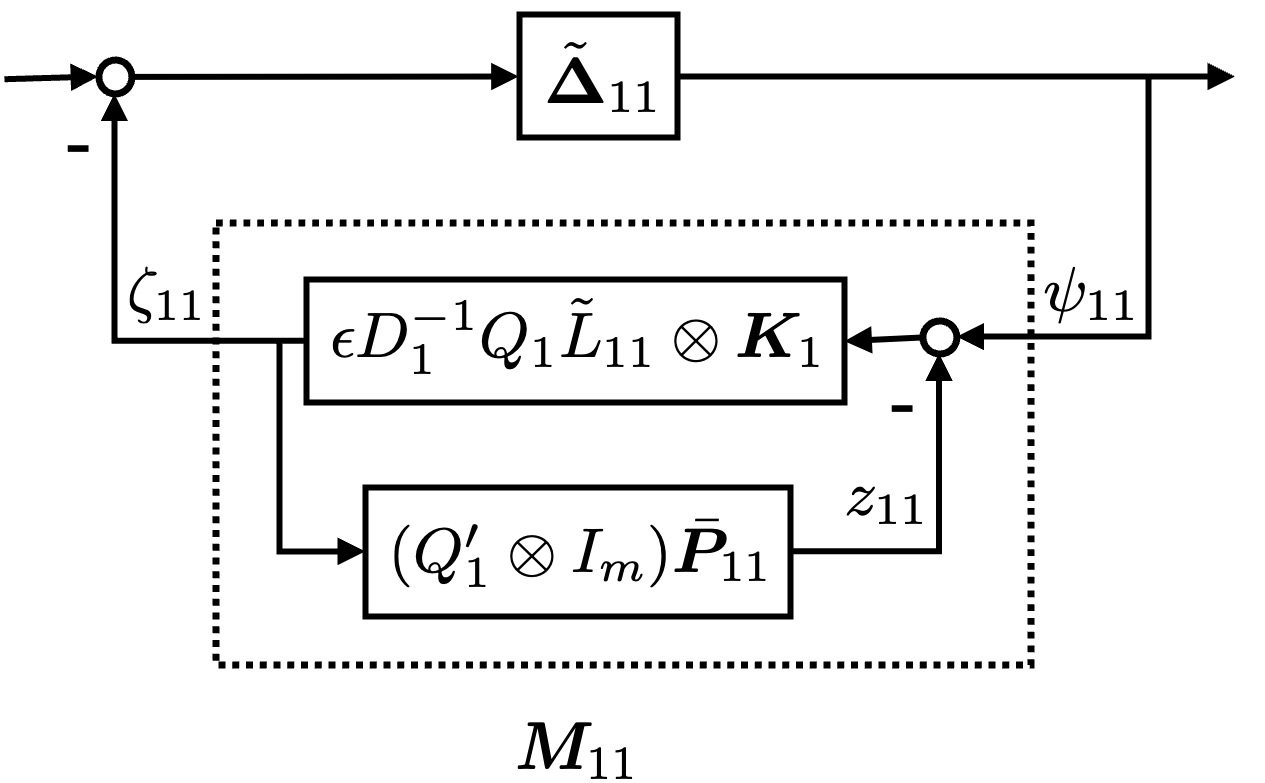} 
                \caption{\label{fig: looptranformation}Loop transformation.}
			\end{center}
		\end{figure}	
        
Next we will show that $\bm M_{11}$ is stable and bounded using state-space method. To this end, we first derive the state-space representation of $\bm M_{11}$. For the sake of notation brevity, we assume that $q=1$, i.e., $e^{j\Omega}=\{1, -1, e^{\omega_1}, e^{-\omega_1}\}$. The proof can be easily extended to the case when $q>1$. 
 Let 
  \begin{align*}
  N_{00}&=\mathrm{diag}\{N_{01},\dots,N_{0r_1}\}, \\
   N_{\pi\pi}&=\mathrm{diag}\{N_{\pi 1},\dots,N_{\pi r_1}\},\\
  N_{11}&=\mathrm{diag}\{N_{11},\dots,N_{1r_1}\}.
  \end{align*}
A minimal realization of $(Q'_1\otimes I_m)\bar{\bm P}_{11}$ is given by $\left[\begin{array}{c|c}
		A_{11}&B_{11}\\\hline
		C_{11}&0
	\end{array}\right]$,  
		where
		\begin{align*}
			A_{11}&=\text{diag}\{1,-1,e^{j\omega_1},e^{-j\omega_1} \}\otimes I_{(r_1-1)m},\\
			B_{11}&=
			\begin{bmatrix}
				(Q'_1\otimes I_m)N_{00} \\
                (Q'_1\otimes I_m)N_{\pi\pi}\\(Q'_1\otimes I_m)N_{11}\\(Q'_1\otimes I_m)
		\bar{N}_{11}	\end{bmatrix}, \\
			C_{11}&=\begin{bmatrix}
				1 & 1 & 1& 1 
			\end{bmatrix} \otimes I_{(r_1-1)m}.
		\end{align*}  
Let $\left[\begin{array}{c|c}
		E_1&F_1\\\hline
		G_1&H_1
	\end{array}\right]$ be a minimal realization of $\bm K_1$, where $E_1\in\mathbb R^{p\times p}$ is stable. Then a minimal realization of $\epsilon D_1^{-1}Q_1\tilde L_{11}\otimes\bm K_1$ is given by $$\left[\begin{array}{c|c}
		I_{r_1-1}\otimes E_1&\epsilon \tilde L_{11}\otimes F_1\\\hline \\[-3mm]D_1^{-1}Q_1\otimes G_1& \epsilon D_1^{-1}Q_1\tilde L_{11}\otimes H_1
	\end{array}\right].$$
        Thus, a minimal realization of $\epsilon(Q'_1\otimes I_m)\bar{\bm P}_{11}(D_1^{-1}Q_1\tilde L_{11}\otimes \bm K_1)$ is given by
        {\small
        {\begin{equation} \label{eq:stateM}
\left[
\begin{array}{c|c}
	\begin{matrix}
		A_{11} & B_{11}(D_1^{-1}Q_1\!\otimes\! G_1)\!\\ 
		0 & I_{r_1-1}\otimes E_1
	\end{matrix} &
	\!\begin{matrix}
		\epsilon B_{11}(D_1^{-1}Q_1\tilde L_{11}\!\otimes\! H_1)\!\\
		\epsilon \tilde L_{11}\otimes F_1
	\end{matrix} 
	\\ \hline \\[-2mm]
	\begin{matrix}
	\hspace{-30pt}	C_{11} \quad\quad\quad & 0
	\end{matrix} &  0
\end{array}
\right].
\end{equation}}}
	
    To transform the state matrix in \eqref{eq:stateM} into a block diagonal matrix, we apply a similarity transformation using the matrix
		\begin{align*}
	\begin{bmatrix}
				I_{4(r_1-1)m} & T     \\
				0 &I_{(r_1-1)p}
			\end{bmatrix},
		\end{align*} where $T$ satisfies the Sylvester equation $$A_{11}T-T(I_{r_1-1}\otimes E_1)=-B_{11}(D_1^{-1}Q_1\otimes G_1).$$ An explicit solution for $T$ is given by
        $$T\!=\!\begin{bmatrix}
				-(Q'_1\otimes I_m)N_{00}(D_1^{-1}Q_1\otimes (G_1(I_p-E_1)^{-1})) \\
                -(Q'_1\otimes I_m)N_{\pi\pi}(D_1^{-1}Q_1\otimes (G_1(-I_p-E_1)^{-1})) \\-(Q'_1\otimes I_m)N_{11}( D_1^{-1}Q_1\otimes (G_1(e^{j\omega_1}I_p-E_1)^{-1}))\\-(Q'_1\otimes I_m)\bar N_{11}(D_1^{-1}Q_1\otimes\! (G_1(e^{-j\omega_1}I_p-E_1)^{-1}))	\end{bmatrix}.$$
From the interpolation condition in constructing $K_1(z)$, we obtain
\begin{align*}
 K_{01}&=H_1+G_1(I_p-E_1)^{-1}F_1,\\
        K_{\pi 1}&=-(H_1+G_1(-I_p-E_1)^{-1}F_1),\\
        K_{11}&=e^{-j\omega_1}(H_1+G_1(e^{j\omega_1}I-E_1)^{-1}F_1).
\end{align*}Thus, after applying a similarity transformation, we obtain an alternative minimal realization of \eqref{eq:stateM}:
      {
		$${\mbox{$\left[
				\begin{array}{c|c}
					\begin{matrix}
						A_{11}     &0\\
						0&I_{r_1-1}\otimes E_1
					\end{matrix}&
					\!\begin{matrix}
						\epsilon \tilde B_{11}\\
			       \epsilon \tilde L_{11}\otimes F_1
					\end{matrix} \!\\ \hline \\[-2mm]
					\begin{matrix}
					C_{11} &\quad C_{11}T
					\end{matrix}&  0
				\end{array}
				\right]$}},$$ } 
                where 
                \begin{align*}\tilde B_{11} =& B_{11}(D_1^{-1}Q_1\tilde L_{11}\otimes H_1)- T(\tilde L_{11}\otimes F_1)\\
                =&\begin{bmatrix}F_{00}\\F_{\pi\pi
                }\\F_{11}\\\bar F_{11}\end{bmatrix}(\tilde L_{11}\otimes I_m),
                \end{align*}
                with
                \begin{align*}
			F_{00}&=(Q_1'\otimes I_m)N_{00}(D_1^{-1}Q_1\otimes K_{01}),\\
            F_{\pi\pi}&=-(Q_1'\otimes I_m)N_{\pi\pi}(D_1^{-1}Q_1\otimes K_{\pi1}),\\
			F_{11}&=(Q_1'\otimes I_m)N_{11}(D_1^{-1}Q_1\otimes e^{j\omega_1}K_{11}).
		\end{align*}
       Denote the output of $(Q_1'\otimes I_m)\bar{\bm P}_{11}$ by $z_{11}$, and the output signal of $\tilde{\bm\Delta}_{11}$ by $\psi_{11}$. Since the input to the block  $\epsilon D_1^{-1}Q_1\tilde L_{11}\otimes\bm K_1$ is equal to $\psi_{11}-z_{11}$, it follows that a minimal realization of $\bm M_{11}$ is given by
\begin{align*}
		\eta(t+1)&=\tilde{A}\eta(t)+\tilde B\psi_{11}(t),\\
			\zeta_{11}(t)&=\tilde C\eta(t)+\tilde D\psi_{11}(t)\nonumber
		\end{align*}
		where $\eta(t)$ and $\zeta_{11}(t)$ are the state and output vector of $\bm M_{11}$ respectively, and
		\begin{align*}
			\tilde{A}=&\!
			\begin{bmatrix}
		A_{11}&0\\
				0&\!\!I_{r_1-1}\otimes E_1
		\end{bmatrix}-\!\epsilon\!		\begin{bmatrix}
				\tilde B_{11}\\
				\tilde{L}_{11}\otimes F_1
			\end{bmatrix}\begin{bmatrix}
				C_{11}& C_{11}T
		\end{bmatrix},\\
\tilde{B}=&\ \epsilon\begin{bmatrix}
				\tilde B_{11}\\
				\tilde{L}_{11}\otimes F_1
			\end{bmatrix},\\
            \tilde{C}=&
			\begin{bmatrix}
				-\epsilon X_cC_{11} &D_1^{-1}Q_1\otimes G_1\!-\!\epsilon X_cC_{11}T
			\end{bmatrix}\!,\\
\tilde{D}=&\ \epsilon D_1^{-1}Q_1\tilde L_{11}\otimes H_1,
	\end{align*}
with 
$X_c=D_1^{-1}Q_1\tilde L_{11}\otimes H_1$.
We will show the stability of $\bm M_{11}$ by finding a Lyapunov function.
 According to \eqref{eq:phasebound}
		and Lemma~\ref{lem:major}, it holds that
		\begin{align*}
			\angle\lambda(F_{00}(\tilde{L}_{11}\otimes I_m))\subset(-\pi/2,\pi/2), \\
            \angle\lambda(-F_{\pi\pi}(\tilde{L}_{11}\otimes I_m))\subset(-\pi/2,\pi/2), \\
			\angle\lambda(e^{-j\omega_1}F_{11}(\tilde{L}_{11}\otimes I_m))\subset(-\pi/2,\pi/2).
		\end{align*}
		Hence there exists $X_0, X_1, X_\pi>0$ such that
		\begin{align*}
			(F_{00}(\tilde{L}_{11}\otimes I_m))^*X_0+X_0(F_{00}(\tilde{L}_{11}\otimes I_m))=I,\\
            -(F_{\pi\pi}(\tilde{L}_{11}\otimes I_m))^*X_\pi-X_\pi(F_{\pi\pi}(\tilde{L}_{11}\otimes I_m))=I,\\
			(e^{-j\omega_1}F_{11}(\tilde{L}_{11}\otimes I_m))^*X_1+X_1(e^{-j\omega_1}F_{11}(\tilde{L}_{11}\otimes I_m))=I.
		\end{align*}
		It follows that
		$$(e^{j\omega_1}\bar{F}_{11}(\tilde{L}_{11}\otimes I_m))^*\bar X_1+\bar X_1(e^{j\omega_1}\bar{F}_{11}(\tilde{L}_{11}\otimes I_m))=I.$$
		Since $I_{r_1-1}\otimes E_1$ is stable, there exists $X_2>0$ such that $$(I_{r_1-1}\otimes E_1)^*X_2(I_{r_1-1}\otimes E_1)-X_2=-Z_1,$$ where $Z_1$ is a positive definite matrix.
  Let
  \begin{align}\label{eq:X}
		X=\begin{bmatrix}
			X_0 &0& \epsilon Y_0& \epsilon Y_1& 0\\
            0 & X_\pi& \epsilon Y_3 & \epsilon Y_4&0\\
			\epsilon Y_0^*& \epsilon Y_3^* & X_1 &\epsilon Y_2& 0\\
			\epsilon Y_1^*& \epsilon Y_4^* &\epsilon Y_2^*&\bar X_1 &0\\
			0& 0 & 0 & 0 &X_2
		\end{bmatrix},
		\end{align}
		with
		\begin{align*}
			Y_0&=\frac{X_0F_{00}(\tilde L_{11}\otimes I_m)+e^{j\omega_1}(F_{11}(\tilde L_{11}\otimes I_m))^*X_1}{e^{j\omega_1}-1},\\
			Y_1&=\frac{X_0F_{00}(\tilde L_{11}\otimes I_m)+e^{-j\omega_1}(\bar{F}_{11}(\tilde L_{11}\otimes I_m))^*\bar X_1}{e^{-j\omega_1}-1},\\
			Y_2&=\frac{X_1F_{11}(\tilde L_{11}\otimes I_m)+(\bar{F}_{11}(\tilde L_{11}\otimes I_m))^*\bar X_1}{e^{-j\omega_1}-e^{j\omega_1}},\\
            Y_3&=\frac{-X_\pi F_{\pi\pi}(\tilde L_{11}\otimes I_m)+e^{j\omega_1}(F_{11}(\tilde L_{11}\otimes I_m))^*X_1}{-e^{j\omega_1}-1},\\
			Y_4&=\frac{-X_\pi F_{\pi\pi}(\tilde L_{11}\otimes I_m)+e^{-j\omega_1}(\bar{F}_{11}(\tilde L_{11}\otimes I_m))^*\bar X_1}{-e^{-j\omega_1}-1}.
		\end{align*}
		Then we have
        {
		\begin{align*}
			\tilde{A}^*X\tilde{A}\!-\!X=\!-\epsilon\left(\begin{bmatrix}
				I_{4(r_1-1)m} & S_1\\
				\!\!S_1^* & \!\!\displaystyle\frac{1}{\epsilon}Z_1+R_1
			\end{bmatrix}\!-\!\epsilon T_1\!-\!\epsilon^2T_2\right),		
		\end{align*}}
		where
		$R_1,S_1,T_1,T_2$ are matrices independent of $\epsilon$. Thus there exists $\epsilon_1\!>\!0$ such that for $\epsilon\!\in\!(0,\epsilon_1)$, it holds that
		\[
		\begin{bmatrix}
			I_{4(r_1-1)m} & S_1\\
			S_1^* & \displaystyle\frac{1}{\epsilon}Z_1+R_1
		\end{bmatrix}>0.
		\]
		Consequently, $\tilde{A}^*X\tilde{A}-X<0$, implying that $\bm M_{11}$ is stable.
		
Next we will prove $\bm M_{11}$ is bounded by showing that $X$, defined in \eqref{eq:X}, satisfies the LMI in the bounded real lemma, i.e., Lemma \ref{lem:brl}.
  Let $\gamma=\|\tilde{\bm \Delta}_{11}\|^{-1}_\infty.$  Suppose $Z_1$ satisfies $$
           Z_1-(I_{r_1-1}\otimes G_1)^*(I_{r_1-1}\otimes G_1)=Z_2>0.$$
 By computation, we obtain
		{
        \begin{align*}
			&\begin{bmatrix}
				\tilde A^*X\tilde A-X+\tilde C^*\tilde C & \tilde A^*X\tilde B+\tilde C^*\tilde D\\
				\tilde B^*X\tilde A+\tilde D^*\tilde C & \tilde B^*X\tilde B-\gamma^2 I+\tilde D^*\tilde D
			\end{bmatrix}=\\
		&-\epsilon\begin{bmatrix}
				\!\begin{bmatrix}
					I & S_2\\
					S_2^* & \!\displaystyle\frac{1}{\epsilon}Z_2\!+\!R_2
		\end{bmatrix}\!\!-\!\epsilon T_3\!-\!\epsilon^2T_4\! 	& \!\!S_3+\epsilon S_4+\epsilon^2S_5\\
			S_3^*+\epsilon S_4^*+\epsilon^2S_5^* & \!\!\displaystyle \frac{\gamma^2}{\epsilon}I\!+\!\epsilon(R_3\!+\!\epsilon^2R_4)
	\end{bmatrix}\!\!,
		\end{align*}}
  where $S_2, S_3, S_4, S_5, R_2, R_3, R_4, T_3, T_4$ are matrices independent of $\epsilon$. For $\epsilon\in(0,\epsilon_1)$, the top-left block of the matrix is positive definite. Moreover, for a fixed gain $\gamma$, there exists $\epsilon_2>0$ such that for $\epsilon\in(0,\epsilon_2)$, the bottom-right block is also positive definite. By analyzing the Schur complement, we conclude that there exists $\epsilon^*\in\text{min}\{\epsilon_1,\epsilon_2\}$ such that the entire block matrix remains positive definite for $\epsilon\in(0,\epsilon^*)$. Consequently, we obtain $\|\bm M_{11}\|_\infty<\gamma$ for $\epsilon\in(0,\epsilon^*)$. 
  
  The stability of $(I+\bm M_{11}\tilde{\bm\Delta}_{11})^{-1}$ follows from the small gain theorem, ensuring the synchronization of the first component. Furthermore, it holds that 
		\begin{align}
			 y_{11}=e_{11}+\mathbf{1}_{r_1} y_0,\label{eq: y1decomp}
		\end{align}
		where $y_0$  represents the synchronized output trajectory in the first component.
		
		We then prove the synchronization in the $j$-th component for $j=2,\dots,\mu$. Without loss of generality, we analyze (\ref{eq: yi*}) for $j=2$, noting that the analysis applies similarly to agents in other components. From Lemma~\ref{lem: laplaciancomb}, we have $-(L_{22}^{-1}L_{21})\mathbf{1}_{r_1}=\mathbf{1}_{r_2}$.
		Substituting (\ref{eq: y1decomp}) into (\ref{eq: yi*}) yields
		\begin{align*}
			 y_{22}=&(I+\epsilon\bm P_{22}(L_{22}\otimes \bm K_2))^{-1}\bm P_{22}(w_{22}\\
			& -\epsilon (I_{r_2}\otimes \bm K_2) (L_{22}L_{22}^{-1}L_{21}\otimes I_m)(e_{11}+\mathbf{1}_{r_1}y_0))\\
			=&(I+\epsilon\bm P_{22}(L_{22}\otimes \bm K_2))^{-1}\bm P_{22}(w_{22}-\epsilon (I_{r_2}\otimes \bm K_2)\\
			&\times \!(L_{21}\!\otimes\! I_m){e}_{11}
			\!+\!\epsilon (I_{r_2}\!\otimes\! \bm K_2)(L_{22}\!\otimes\! I_m)(\mathbf{1}_{r_2}\!\otimes\! {y}_0)).
		\end{align*}
		It suffices to show that the agents in $\mathcal P_2$ can achieve synchronization with the output trajectory ${y}_{0}$, which is equivalent to demonstrating that $(I+\epsilon\bm P_{22}(L_{22}\otimes \bm K_2))^{-1}$ is internally stable, the semi-stable poles of $\bm P_{22}$ are cancelled by the zeros of $(I+\epsilon\bm P_{22}(L_{22}\otimes \bm K_2))^{-1}$ and the output signal $ y_{22}$ can track the signal $\mathbf{1}_{r_2}\otimes  y_{0}$, as illustrate in Fig.~\ref{fig: tracking}.
		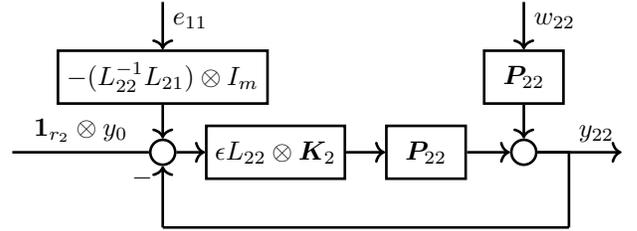
\begin{figure}[htbp]
			\centering
			\tikzstyle{block} = [draw, rectangle,
			minimum height=2em, minimum width=3em]
			\tikzstyle{sum} = [draw, circle, node distance=2.2cm]
			\tikzstyle{input} = [coordinate]
			\tikzstyle{output} = [coordinate]
			\tikzstyle{divide} = [coordinate]
			\tikzstyle{pinstyle} = [pin edge={to-,thin,black}]
			
			\begin{tikzpicture}[auto]
				\node [divide, name=R, node distance=1cm] (R) {};
				\node [sum,line width=0.35mm, above of=R,node distance=1cm] (divide1) {};
				\node [divide, left of=divide1,node distance=2cm] (input) {};
				\node [block, line width=0.35mm,right of=divide1, node distance=1.5cm] (H) {$\epsilon L_{22} \otimes \bm K_2$};
				\node [block,line width=0.35mm, right of=H, node distance=2cm] (G) {$\bm P_{22}$};
				\node [sum, line width=0.35mm, right of=G,node distance=1.3cm] (sum1) {};
				\node [block, line width=0.35mm, above of=sum1,node distance=1cm] (K) {$\bm P_{22}$};
				\node [block, line width=0.35mm, above of=divide1,node distance=1cm] (L21) {$-(L_{22}^{-1}L_{21})\otimes I_m$};
				\node[divide, above of=K,node distance=1cm](x){};
				\node[divide, above of=L21,node distance=1cm](ydis){};
				\node [divide, right of=R, node distance=5.4cm] (T) {};
				\node[divide, right of=sum1,node distance=0.8cm](divide3){};
				\node[divide, right of=divide3,node distance=0.5cm](end){};
				\draw [line width=0.35mm,->] (R) -- node[pos=0.8]{$-$}node {} (divide1);
				\draw [line width=0.35mm,-] (input) -- node[pos=0.5]{$\mathbf{1}_{r_2}\otimes{y}_0$} node{}(divide1);
				\draw [line width=0.35mm,->] (divide1) -- node {} (H);
				\draw [line width=0.35mm,->] node {} (H)-- node {} (G);
				\draw [line width=0.35mm,-] (sum1) -|  node{}(T);
				\draw [line width=0.35mm,->] (G) -- node {} (sum1);
				\draw [line width=0.35mm,-]  (T) |-(R);
				\draw [line width=0.35mm,->] (x) -- node[pos=0.4]{$w_{22}$} node{}(K);
				\draw [line width=0.35mm,->] (ydis) -- node[pos=0.4]{$e_{11}$} node{}(L21);
				\draw [line width=0.35mm,->] (K) -- node {} (sum1);
				\draw [line width=0.35mm,->] (L21) -- node {} (divide1);
				\draw [line width=0.35mm,->] (sum1) --  node [pos=0.7]{$ y_{22}$}(end);
			\end{tikzpicture}
			\caption{\label{fig: tracking}Block diagram of a feedback stability and tracking problem.}
		\end{figure}
		
		Assume that $D_2$ is an optimal diagonal scaling matrix corresponding to $\phi_{\mathrm{ess}}(L_{22})$. That $(I+\epsilon\bm P_{22}(L_{22}\otimes \bm K_2))^{-1}$ is stable is equivalent to $(I+\epsilon\tilde{\bm P}_{22}(D_2^{-1}L_{22}D_2)\otimes I_m)^{-1}$ is stable,
		where
		$$\tilde{\bm P}_{22}=(D_2^{-1}\otimes I_m)\bm P_{22}(I_{r_2}\otimes \bm K_2)(D_2\otimes I_m).$$ The stability of $(I+\epsilon\tilde{\bm P}_{22}(D_2^{-1}L_{22}D_2)\otimes I_m)^{-1}$ can then be established using a similar approach as for $\bm S_{11}$. Since the agents dynamic $\bm P_{22}$ inherently includes its internal model \cite{bengtsson1977output}, the semi-stable poles of $\bm P_{22}$ are cancelled by the zeros of $(I+\epsilon\bm P_{22}(L_{22}\otimes \bm K_2))^{-1}$.
		Moreover, the agents in $\bm P_{22}$ share the same poles on the unit circle $e^{j\Omega}$ with the input signal $\mathbf{1}_{r_2}\otimes y_0$, according to the internal model principle \cite{francis1976internal}, the outputs signal $y_{22}$ is able to track $\mathbf{1}_{r_2}\otimes y_0$ without steady state error. The theorem is proved.		
	\end{proof}

\end{document}